\documentclass[%
 reprint,
superscriptaddress,
 amsmath,amssymb,
 aps,
pra,
]{revtex4-1}

\usepackage{amsthm,amssymb}
\usepackage{amsmath}
\usepackage{graphicx}
\usepackage{dcolumn}
\usepackage{bm}
\usepackage{multirow} 
\usepackage{mathdots}
\usepackage{enumerate}
\usepackage{algorithm}
\usepackage{algorithmicx}
\usepackage{algpseudocode}
\usepackage{textcomp}
\usepackage{color}

\usepackage{epstopdf}
\usepackage{epsfig}
\usepackage{array}
\usepackage{booktabs}

\theoremstyle{remark}
\newtheorem{definition}{\indent Definition}
\newtheorem{lemma}{\indent Lemma}
\newtheorem{theorem}{\indent \emph{\textbf{Theorem}}}

\begin{document}

\makeatletter
\newcommand{\ud}{\mathrm{d}}
\newcommand{\rmnum}[1]{\romannumeral #1}
\newcommand{\polylog}{\mathrm{polylog}}
\newcommand{\ket}[1]{|{#1}\rangle}
\newcommand{\bra}[1]{\langle{#1}|}
\newcommand{\inn}[2]{\langle{#1}|#2\rangle}
\newcommand{\Rmnum}[1]{\expandafter\@slowromancap\romannumeral #1@}
\newcommand{\udots}{\mathinner{\mskip1mu\raise1pt\vbox{\kern7pt\hbox{.}}
        \mskip2mu\raise4pt\hbox{.}\mskip2mu\raise7pt\hbox{.}\mskip1mu}}
\makeatother

\preprint{APS/123-QED}

\title{A general quantum matrix exponential dimensionality reduction framework based on block-encoding}

\author{Yong-Mei Li}
\affiliation{State Key Laboratory of Networking and Switching Technology, Beijing University of Posts and Telecommunications, Beijing, 100876, China}
\author{Hai-Ling Liu}
\affiliation{State Key Laboratory of Networking and Switching Technology, Beijing University of Posts and Telecommunications, Beijing, 100876, China}
\author{Shi-Jie Pan}
\affiliation{State Key Laboratory of Networking and Switching Technology, Beijing University of Posts and Telecommunications, Beijing, 100876, China}
\author{Su-Juan Qin}
\email{qsujuan@bupt.edu.cn}
\affiliation{State Key Laboratory of Networking and Switching Technology, Beijing University of Posts and Telecommunications, Beijing, 100876, China}
\author{Fei Gao}
\email{gaof@bupt.edu.cn}
\affiliation{State Key Laboratory of Networking and Switching Technology, Beijing University of Posts and Telecommunications, Beijing, 100876, China}
\author{Qiao-Yan Wen}
\affiliation{State Key Laboratory of Networking and Switching Technology, Beijing University of Posts and Telecommunications, Beijing, 100876, China}

\date{\today}

\begin{abstract}
As a general framework, Matrix Exponential Dimensionality Reduction (MEDR) deals with the small-sample-size problem that appears in linear Dimensionality Reduction (DR) algorithms. High complexity is the bottleneck in this type of DR algorithm because one has to solve a large-scale matrix exponential eigenproblem. To address it, here we design a general quantum algorithm framework for MEDR based on the block-encoding technique. This framework is configurable, that is, by selecting suitable methods to design the block-encodings of the data matrices, a series of new efficient quantum algorithms can be derived from this framework. Specifically, by constructing the block-encodings of the data matrix exponentials, we solve the eigenproblem and then obtain the digital-encoded quantum state corresponding to the compressed low-dimensional dataset, which can be directly utilized as input state for other quantum machine learning tasks to overcome the curse of dimensionality. As applications, we apply this framework to four linear DR algorithms and design their quantum algorithms, which all achieve a polynomial speedup in the dimension of the sample over their classical counterparts.
\end{abstract}

\pacs{Valid PACS appear here}
\maketitle

\section{Introduction}

The power of quantum computing is illustrated by quantum algorithms that solve specific problems, such as factoring~\cite{PS1994}, unstructured data search~\cite{GL1996}, linear systems~\cite{HHL} and cryptanalysis~\cite{lizhenqiang}, much more efficiently than classical algorithms. In recent years, a series of quantum algorithms for solving machine learning problems have been proposed and attracted the attention of the scientific community, such as clustering~\cite{Lloyd2013,Otterbach2017,Kerenidis2021}, dimensionality reduction \cite{QPCA,QLDA,QAOP,QNPE}, and matrix computation~\cite{wan2018,liu,wan2021}. Quantum Machine Learning (QML)~\cite{QML} has appeared as a remarkable emerging direction with great potential in quantum computing.

In the era of big data, Dimensionality Deduction (DR) has gained significant importance in machine learning and statistical analysis, which is a powerful technique to reveal the intrinsic structure of data and mitigate the effects of the curse of dimensionality~\cite{Hammer1961,PRML2006}. The essential task of DR is to find a mapping function $F:{\bf x}\mapsto {\bf y}$ that transforms ${\bf x}\in \mathbb{R}^M$ into the desired low-dimensional representation ${\bf y}\in \mathbb{R}^m$, where, typically, $m\ll M$. Most of the DR algorithms, such as Locality Preserving Projections (LPP)~\cite{LPP}, Unsupervised Discriminant Projection (UDP)~\cite{UDP}, Neighborhood Preserving Embedding (NPE)~\cite{NPE}, and Linear Discriminant Analysis (LDA)~\cite{LDA}, were unified into a general graph embedding framework proposed by Yan et al.~\cite{GE2007}. The framework provides a unified perspective for the understanding and comparison of many popular DR algorithms and facilitate the design of new algorithms. Unfortunately, almost all linear DR algorithms in the graph embedding framework encounter the well known Small-Sample-Size (SSS) problem that stems from the generalized eigenproblems with singular matrices. To tackle it, a general Matrix Exponential Dimensionality Reduction (MEDR) framework~\cite{GME2014} is proposed. In the framework, the SSS problem is solved by transforming the generalized eigenproblem into the eigenproblem $e^{-S_2}e^{S_1}{\bf v}=\lambda {\bf v}$, where the data matrices $S_1$ and $S_2$ have different forms for different algorithms. However, it involves the solution of the large-scale matrix exponential eigenproblem, resulting in the matrix-exponential-based DR algorithms being time-consuming for a large number of high-dimensional samples. Therefore, it would be of great significance to seek new strategies to speedup this type of DR algorithm.

There exists some work on quantum DR algorithms, which achieve different degrees of acceleration compared with classical algorithms~\cite{QPCA,QLDA,Yu2018QPCA,QLLE2019,QAOPDuan,QAOP,QNPEliang,Li2021nDR,PRA2021nDR,QLPP, QNPE,liDCCA,YU2023QLDA}. In the early stage, the work is focused on linear DR (i.e., $F$ is a linear function). Lloyd et al.~\cite{QPCA} proposed the first quantum DR algorithm, quantum Principal Component Analysis (PCA), which provides an important reference for many subsequent quantum algorithms. Later, the quantum nonlinear DR algorithms gradually emerged based on  manifold learning~\cite{QLLE2019,PRA2021nDR} and kernel method~\cite{Li2021nDR}. However, there is no quantum algorithm that efficiently realizes the matrix-exponential-based DR. Then an interesting question is whether we can design quantum algorithms for this type of DR algorithm in a unified way, which will provide computational benefits and facilitate the design of new quantum DR algorithms. We answer the question in the affirmative by using the method of block-encoding~\cite{ICALP2019SDP,Low2019,GSL,ICALP2019}.

Block-encoding is a good framework for implementing matrix arithmetic on quantum computers. The block-encoding $U$ of a matrix $A$ is a unitary matrix whose top-left block is proportional to $A$. Given $U$, one can produce the state $A|\phi\rangle/\|A|\phi\rangle\|$ by applying $U$ to an initial state $|0\rangle|\phi\rangle$. Low and Chuang~\cite{Low2019} showed how to perform optimal Hamiltonian simulation given a block-encoded Hamiltonian $H$. Based on this, Chakraborty et al.~\cite{ICALP2019} developed several tools within the block-encoding framework, such as singular value estimation, and quantum linear system solvers. Moreover, the block-encoding technique has been used to design QML algorithms, such as quantum classification~\cite{Shao2020} and quantum DR~\cite{liDCCA}.

In this paper, we apply the block-encoding technique to MEDR and design a general Quantum MEDR (QMEDR) framework. This framework is configurable, that is, for a matrix-exponential-based DR algorithm, one can take suitable methods to design the block-encodings of $S_1$ and $S_2$ and then construct the corresponding quantum algorithm to obtain the compressed low-dimensional dataset. Once these two block-encodings are implemented efficiently, the quantum algorithm will have a better running time compared to its classical counterpart. More specifically, the main contributions of this paper are as follows.

(a) Given a block-encoded Hermitian $H$, we present a method to implement the block-encoding of $e^H$ or $e^{-H}$ and derive the error upper bound. It provides a way for the computation of the matrix exponential which is one of the most important tasks in linear algebra~\cite{Higham2008}. To be specific, depending on the application, the computation of the matrix exponential may be to compute $e^A$ for a given square matrix $A$, to apply $e^A$ to a vector, and so on~\cite{Higham2008}. They are very time-consuming when $A$ is an exponentially large matrix. With the help of block-encodings, these tasks can be performed much faster on a quantum computer.

(b) By combining the block-encoding technique and quantum phase estimation~\cite{NC}, we solve the eigenproblem $e^{-S_2}e^{S_1}{\bf v}=\lambda {\bf v}$ and then construct the compressed digital-encoded state~\cite{QADC2019}, i.e., directly encode the compressed low-dimensional dataset into qubit strings in quantum parallel. The compressed state can be further utilized as input state for varieties of QML tasks, such as quantum $k$-medoids clustering~\cite{lym}, to overcome the curse of dimensionality. This builds a bridge between quantum DR algorithms and other QML algorithms. In addition, the proposed method for constructing the compressed state can be extended to other linear DR algorithms.

(c) As applications, we apply the QMEDR framework to four matrix-exponential-based DR algorithms, i.e., Exponential LPP (ELPP)~\cite{ELPP}, Exponential UDP (EUDP)~\cite{GME2014}, Exponential NPE (ENPE)~\cite{ENPE}, and Exponential Discriminant Analysis (EDA)~\cite{EDA}, and design their quantum algorithms. The results show that all of these quantum algorithms achieve a
polynomial speedup in the dimension of the sample over the classical algorithms. Moreover, for the number of samples, the quantum ELPP and quantum ENPE algorithms achieve a polynomial speedup over their classical counterparts and the quantum EDA algorithm provides an exponential speedup.

The remainder of the paper proceeds as follows. In Sec.~\ref{sec:Pre}, we review the MEDR framework in Sec.~\ref{subsec:MEDR} and review the block-encoding framework in Sec.~\ref{subsec:BK}. In Sec.~\ref{sec:QAF}, we propose the QMEDR framework in Sec.~\ref{sub:alg} and analyze its complexity in Sec.~\ref{sub:ca}. In Sec.~\ref{Sec:exm}, we present some applications of the QMEDR framework. In Sec.~\ref{Sec:dis}, we discuss the main idea of the QMEDR framework and compare the framework with the related work. The conclusion is given in Sec.~\ref{Sec:con}.

\section{Preliminaries}\label{sec:Pre}

\subsection{MEDR framework}\label{subsec:MEDR}

In this subsection, we first review the linear DR in the view of graph embedding and then review the general MEDR framework and analyze its complexity.

Let $X=[{\bf x}_0,{\bf x}_1,...,{\bf x}_{N-1}]^T\in \mathbb{R}^{N\times M}$ denote the original data matrix. DR aims to seek an optimal transformation to map the $M$-dimensional sample ${\bf x}_i$ onto a $m$-dimensional ($m\ll M$) sample ${\bf y}_i$, $i=0,1,\cdots,N-1$. Many DR algorithms were unified into a general graph embedding framework~\cite{GE2007}. In the framework, an undirected weighted graph $G=\{\mathcal{X},S\}$ with vertex set $\mathcal{X}=\{{\bf x}_i\}_{i=0}^{N-1}$ and similarity matrix $S=(S_{ij})\in\mathbb{R}^{N\times N}$ is defined to characterize
the original dataset, where $S_{ij}$ measures the similarity of a pair of vertices ${\bf x}_i$ and ${\bf x}_j$. The graph-preserving criterion is
\begin{align}\label{GPC}
{\bf y}^{*}=\mathop{\arg\min}\limits_{{\bf y}^TB{\bf y}=d} \sum_{i\neq j}\|{\bf y}_i-{\bf y}_j\|^2S_{ij}=\mathop{\arg\min}\limits_{{\bf y}^TB{\bf y}=d}{\bf y}^TL{\bf y},
\end{align}
where $B$ is a constraint matrix, $d$ is a constant, $L=D-S$ is the Laplacian matrix, $D$ is a diagonal matrix, and  $D_{ii}=\sum_{j\neq i}S_{ij}$.
Note that $\|\cdot\|$ is the $l_2$-norm of a vector or the spectral norm of a matrix in this paper.

Let ${\bf y}_i=W^T{\bf x}_i $ be the linear mapping from the original sample ${\bf x}_i$ onto the desired low-dimensional representation ${\bf y}_i$. Then Eq.(\ref{GPC}) can be reformulated as
\begin{align}\label{equ:obj}
W^{*}=\mathop{\arg\min}\limits_{W} \bigg|\big(W^TS_2W\big)^{-1}\big(W^TS_1W\big)\bigg|,
\end{align}
where $S_1=L$ or $X^TLX$, and $S_2=I$, $B$, or $X^TBX$. 

The solutions of Eqs.(\ref{GPC}) and (\ref{equ:obj}) are obtained by solving the generalized eigenproblem
\begin{align}\label{GE}
S_1{\bf v}=\lambda S_2{\bf v}.
\end{align}
Let ${\bf v}_0,{\bf v}_1,\cdots,{\bf v}_{m-1}$ be the generalized eigenvectors corresponding to the first $m$ smallest generalized eigenvalues. Then $W=[{\bf v}_0,{\bf v}_1,\cdots,{\bf v}_{m-1}]$ is the optimal transformation matrix and the desired low-dimensional data matrix $Y=XW=[{\bf y}_0,{\bf y}_1,\cdots,{\bf y}_{N-1}]^T$.

However, in many cases of real life, $N< M$, resulting in the matrices $S_1$ and $ S_2$ being singular and then Eq.(\ref{GE}) is unsolvable. This is a so-called SSS problem. To solve it, the general MEDR framework is proposed~\cite{GME2014}. It replaced the matrices $S_1$ and $S_2$ in Eq.(\ref{GE}) with their matrix exponentials~\cite{Higham2008} respectively. That is to say, the MEDR framework should solve the eigenproblem
\begin{align}\label{equ:Ess}
e^{-S_2}e^{S_1}{\bf v}=\lambda {\bf v}
\end{align}
and take the first $m$ principal eigenvectors to obtain the optimal transformation matrix. The flowchart of the MEDR framework is shown in FIG.~\ref{figure1}.

\begin{figure}[htbp!]
	
	\begin{minipage}{1\linewidth}
		\vspace{3pt}
		\centerline{\includegraphics[width=1\textwidth]{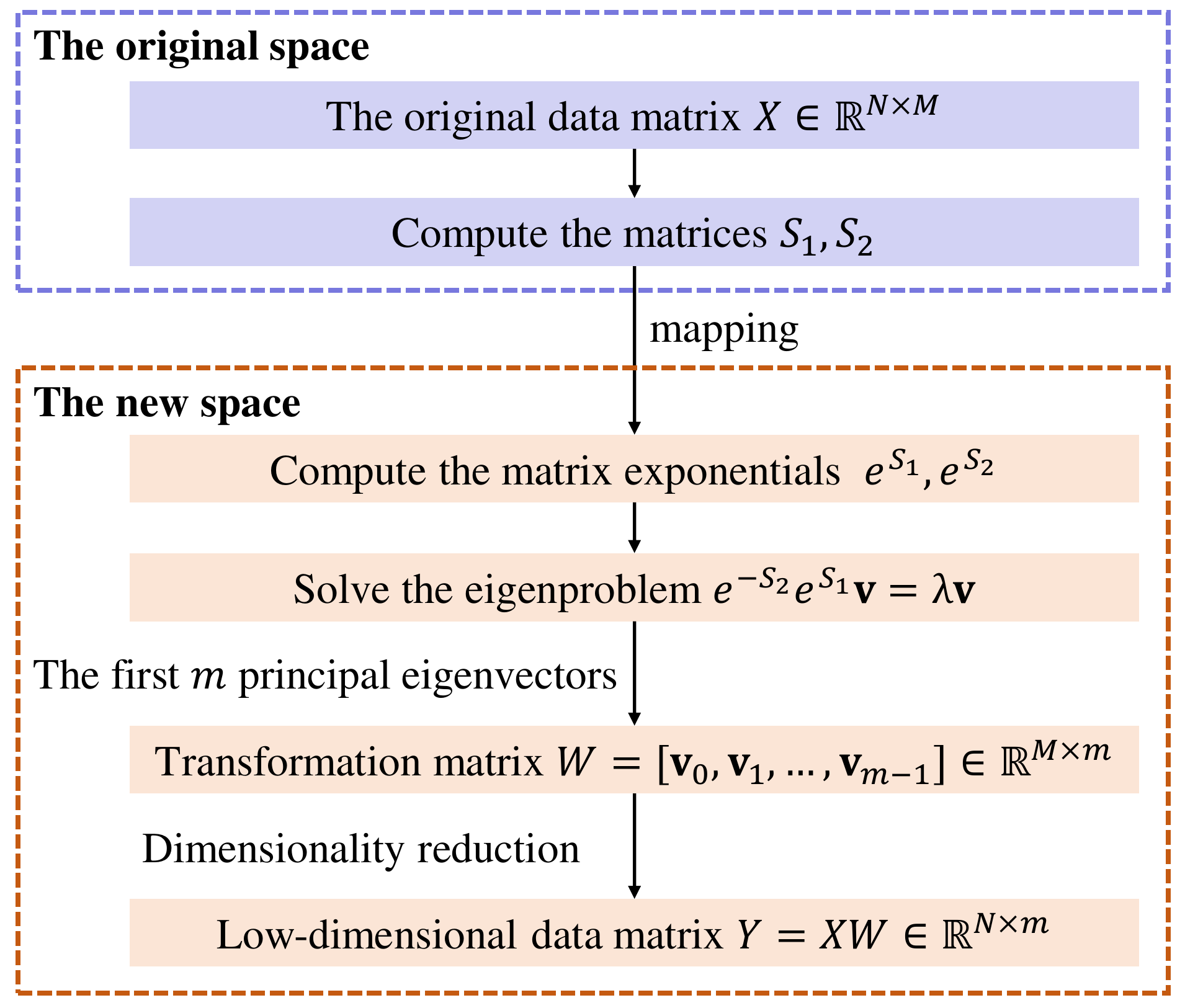}}
	\end{minipage}
	
	\caption{Flowchart of the general MEDR framework.}
	\label{figure1}
\end{figure}

The matrix exponential method solves the SSS problem well. Nevertheless, its high complexity constitutes the bottleneck in this type of matrix-exponential-based algorithm because one has to compute the multiplications of matrices and the exponentials of matrices and to solve a matrix exponential eigenproblem. First, for a matrix in the type of  $X^TFX$, the time complexity of computing it is  $O(N^2M+NM^2)$, where $F$ is a square matrix of order $N$. Note that although $N$ is less than $M$ and can be ignored, we keep $N$ here for a clearer representation of the parameter relationship. Second, in general, for a given $M$-by-$M$ matrix, the complexity of computing its matrix exponential is $O(M^3)$~\cite{matrixE,Higham2008}. Third, the complexity of solving the eigenproblem is $O(M^3)$ for a square matrix of order $M$. Add them up, the total time complexity is $O(N^2M+M^3)$ which brings difficulty to the practical applications of this type of algorithm when dealing with a larger number of high-dimensional data. Therefore, it is necessary to seek new strategies to speedup the matrix-exponential-based DR algorithms.

\subsection{Block-encoding framework}\label{subsec:BK}

In this subsection, we briefly review the framework of block-encoding. For convenience, we use $\log (\cdot )$ to denote $\log_2(\cdot)$ throughout the paper.

\begin{definition}(block-encoding~\cite{GSL}).
Suppose that $A$ is an $s$-qubit operator, $\alpha,\epsilon \in \mathbb{R}_+$, and $a\in \mathbb{N}$, then we say that the $(s+a)$-qubit unitary $U$ is an $(\alpha, a,\epsilon)$-block-encoding of $A$, if
\begin{align}
\|A-\alpha(\langle0|^{\otimes a}\otimes I)U(|0\rangle^{\otimes a}\otimes I)\|\leq \epsilon.
\end{align}
\end{definition}

Note that since $\|U\|=1$, we necessarily have $\|A\|\leq\alpha+\epsilon$. Moreover, the above definition is not only restricted to square matrices. When $A$ is not a square matrix, we can define an embedding square matrix $A_e$ such that the top-left block of $A_e$ is $A$ and all other elements are 0.

There are some methods to implement the block-encodings for specific matrices, such as sparse matrices~\cite{GSL,Low2019}, density operators~\cite{ICALP2019SDP,Low2019,GSL}, Gram matrices~\cite{GSL}, and matrices stored in structured Quantum Random Access Memories (QRAMs)~\cite{CGJ,ICALP2019,KP2020}.

The main motivation for using block-encodings is to perform optimal-Hamiltonian simulation for a given block-encoded
Hamiltonian.

\begin{theorem}(Optimal block-Hamiltonian simulation~\cite{LC2016,CGJ,SDP}).\label{the:OBC}
Suppose that $U$ is an $(\alpha, a, \frac{\epsilon}{|2t|})$-block-encoding of the Hamiltonian $H$. Then we can implement an $\epsilon$-precise Hamiltonian
simulation unitary $V$ which is an $(1, a + 2, \epsilon)$-block-encoding of $e^{iHt}$, with $O(|\alpha t|+\frac{\log(1/\epsilon)}{\log\log(1/\epsilon)})$ uses of controlled-$U$ or its inverse and with $O(a|\alpha t|+a\frac{\log(1/\epsilon)}{\log\log(1/\epsilon)})$ two-qubit gates.
\end{theorem}

Furthermore, the block-encoding technique has been applied to many quantum algorithms, such as quantum linear systems~\cite{Low2019,ICALP2019,wan2021}, and quantum mean centering~\cite{liuMC}. In the following section, we will apply the block-encoding framework to MEDR and design a general quantum algorithm framework for MEDR to speedup the matrix-exponential-based DR algorithms.

\section{General QMEDR framework based on block-encoding}\label{sec:QAF}

In this section, we will detail the general quantum algorithm framework that performs matrix-exponential-based DR more efficiently than what is achievable classically under certain conditions. For simplicity, let $\widetilde{O}(\cdot)$ symbolize algorithm complexity where parameters with polylogarithmic dependence are hidden.

Assume that the original data matrix $X=[{\bf x}_0,{\bf x}_1,...,{\bf x}_{N-1}]^T\in \mathbb{R}^{N\times M}$ is stored in a structured QRAM~\cite{QRAMDS1,QRAMDS2,ICALP2019}. Then there exists a quantum algorithm that can perform
the mapping
\begin{align}
\textbf{O}_1: |i\rangle|0\rangle\mapsto |i\rangle|\|{\bf x}_i\|\rangle
\end{align}
in time $O(\text{poly}\log(NM))$ and perform
the following mappings with $\varepsilon_x$-precision in time $O(\text{poly}\log(\frac{NM}{\varepsilon_x}))$:
\begin{align}
&\textbf{O}_2: |i\rangle|0\rangle\mapsto |i\rangle|{\bf x}_i\rangle=|i\rangle\frac{1}{\|{\bf x}_i\|}\sum_{j=0}^{M-1}x_{ij}|j\rangle,\nonumber
\\&\textbf{O}_3: |0\rangle|j\rangle\mapsto \frac{1}{\|X\|_F}\sum_{i=0}^{N-1}\|{\bf x}_i\||i\rangle|j\rangle,
\end{align}
where $x_{ij}$ denotes the $(i,j)$-entry of $X$ and $\|\cdot\|_F$ is the Frobenius norm for a matrix.

Based on this assumption, the framework can be summarized as the following theorem.

\begin{theorem}(QMEDR framework).\label{QMEDR}
For an algorithm belonging to the MEDR framework, let $\kappa_1,\kappa_2\geq 2$ such that $\frac{I}{\kappa_1}\preceq S_1\preceq I$, $\frac{I}{\kappa_2}\preceq S_2\preceq I$. Suppose that $U_1$ is an $(\alpha, a, \varepsilon)$-block-encoding of $S_1$ and $U_2$ is an $(\beta, b, \delta)$-block-encoding of $S_2$, which can be implemented in times $T_1$ and $T_2$ respectively. Then there exists a quantum algorithm to produce the compressed digital-encoded state
\begin{align}
\frac{1}{\sqrt{Nm}}\sum_{i=0}^{N-1}\sum_{j=0}^{m-1}|i\rangle|j\rangle|y_{ij}\rangle :=|\psi\rangle
\end{align}
in time $\widetilde{O}(\frac{(T+a+b)\max_i\|{\bf x}_i\|^2m\sqrt{M}}{\epsilon})$, where $T=\widetilde{O}\big(\max\{\alpha\kappa_1(a+T_1),\beta\kappa_2(b+T_2)\}\big)$, $y_{ij}$ is the $(i,j)$-entry of the low-dimensional data matrix $Y$ and its error is $\epsilon$.
\end{theorem}

To prove Theorem~\ref{QMEDR}, we give the framework details in Sec.~\ref{sub:alg} and analyze its complexity in Sec.~\ref{sub:ca}.

\subsection{QMEDR framework}\label{sub:alg}

The core of the QMEDR framework is to extract the first $m$ principal eigenvectors of $e^{-S_2}e^{S_1}$ and then construct the compressed digital-encoded state. To achieve it, we first use the block-encoding technique to simulate $e^{-S_2}e^{S_1}$, which can be done by constructing its block-encoding. With it, we extract the first $m$ smallest eigenvalues and the associated eigenvectors by combining quantum phase estimation~\cite{NC} with the quantum minimum-finding algorithm~\cite{1996min}. Next, we construct the compressed digital-encoded quantum state by inner product estimation~\cite{qmeans2019}. The QMEDR framework consists of three steps, we now detail them one by one.

\subsubsection{Construct the block-encoding of $e^{-S_2}e^{S_1}$}

We first give the following lemma, which is inspired by~\cite{CGJ} and allows us to construct the block-encodings of $e^{S_1}$ and $e^{-S_2}$.

\begin{lemma}(Implementing the block-encodings of matrix exponential and its inverse).\label{lem:expH}
Let $c\in\{1,-1\}$, $\epsilon\in(0,\frac{1}{2}]$, $\kappa\geq 2$ and $H$ be a Hermitian matrix such that $\frac{I}{\kappa}\preceq H\preceq I$. If for some fixed $\omega\in \mathbb{R}_+$ we have $\delta\leq\omega\epsilon/\big(\kappa\log(\frac{1}{\epsilon})\log^2(\kappa\log(\frac{1}{\epsilon}))\big)$
and $U$ is an $(\alpha,a,\delta)$-block-encoding of $H$ which can be implemented in time $T_U$, then we can implement a unitary $\tilde{U}$ that is an $(e^2,a+O(\log(\kappa\log(\frac{1}{\epsilon}))),e^2\epsilon)$-block-encoding of $e^{cH}$ in cost
$O\big(\alpha\kappa\log(\frac{1}{\epsilon})(a+T_U)+\kappa\log(\frac{\kappa}{\epsilon})\log(\frac{1}{\epsilon})\big)$.
\end{lemma}
\begin{proof}
See Appendix~\ref{app:proofexpH}.
\end{proof}

Note that the above lemma is not only restricted to Hermitian matrices. For a non-Hermitian matrix, one can extend it to an embedding Hermitian matrix by Lemma~\ref{lem:2H}.

Given $U_1$ and $U_2$ which are the block-encodings of $S_1$ and $S_2$ respectively, we can implement an $(e^2,a+O(\log(\kappa_1\log(\frac{1}{\epsilon_1}))),e^2\epsilon_1)$-block-encoding of $e^{S_1}$ and an $(e^2,b+O(\log(\kappa_2\log(\frac{1}{\epsilon_2}))),e^2\epsilon_2)$-block-encoding of $e^{-S_2}$ by Lemma~\ref{lem:expH}, where $\epsilon_1,\epsilon_2\in (0,\frac{1}{2}]$. Then an $(e^4,a+b+O(\log(\kappa_1\log(\frac{1}{\epsilon_1}))+\log(\kappa_2\log(\frac{1}{\epsilon_2}))),e^4(\epsilon_1+\epsilon_2))$-block-encoding of $e^{-S_2}e^{S_1}$ is constructed by product of block-encoded matrices~\cite{GSL}. Based on this, we can simulate $e^{-S_2}e^{S_1}$ by Theorem~\ref{the:OBC}.

\subsubsection{Extract the first $m$ principal eigenvalues}

Here we first use quantum phase estimation~\cite{NC} to reveal the eigenvalues and eigenvectors of $e^{-S_2}e^{S_1}$, then use the quantum minimum-finding algorithm~\cite{1996min} to search the first $m$ smallest eigenvalues. The details are as follows.

(2.1) Reveal eigenvalues and eigenvectors of $e^{-S_2}e^{S_1}$.

For simplicity, we assume that the matrix $e^{-S_2}e^{S_1}$ is Hermitian. If not, we can extend it to an embedding Hermitian matrix, see Appendix~\ref{app:notH2}. Given the block-encoding of $e^{-S_2}e^{S_1}$, the unitary $e^{\imath (e^{-S_2}e^{S_1})t}$ can be implemented according to Theorem~\ref{the:OBC}, where $\imath^2=-1$. By using it, we apply quantum phase estimation on the first two registers of the state
\begin{align}\label{equ:ent}
|0\rangle^{\otimes q_1}\frac{1}{\sqrt{M}}\sum_{i=0}^{M-1}|i\rangle|i\rangle,
\end{align}
then we obtain the state
 \begin{align}\label{equ:sta1}
\frac{1}{\sqrt{M}}\sum_{i=0}^{M-1}|\lambda_i\rangle|{\bf v}_i\rangle|{\bf v}_i\rangle :=|\Psi\rangle,
\end{align}
where $\lambda_i$ and $|{\bf v}_i\rangle$ are the eigenvalue and the corresponding eigenvector of $e^{-S_2}e^{S_1}$, the value of $q_1$ determines the accuracy of phase estimation and we discuss it in Sec.~\ref{sub:ca}. Such entangled state in Eq.(\ref{equ:ent}) can be efficiently constructed by using Hadamard gates and CNOT gates, as shown in FIG.~\ref{figure2}, and its partial trace over the first two registers is proportional to the identity matrix.

\begin{figure}[htbp!]
	
	\begin{minipage}{1\linewidth}
		\vspace{3pt}
		\centerline{\includegraphics[width=0.95\textwidth]{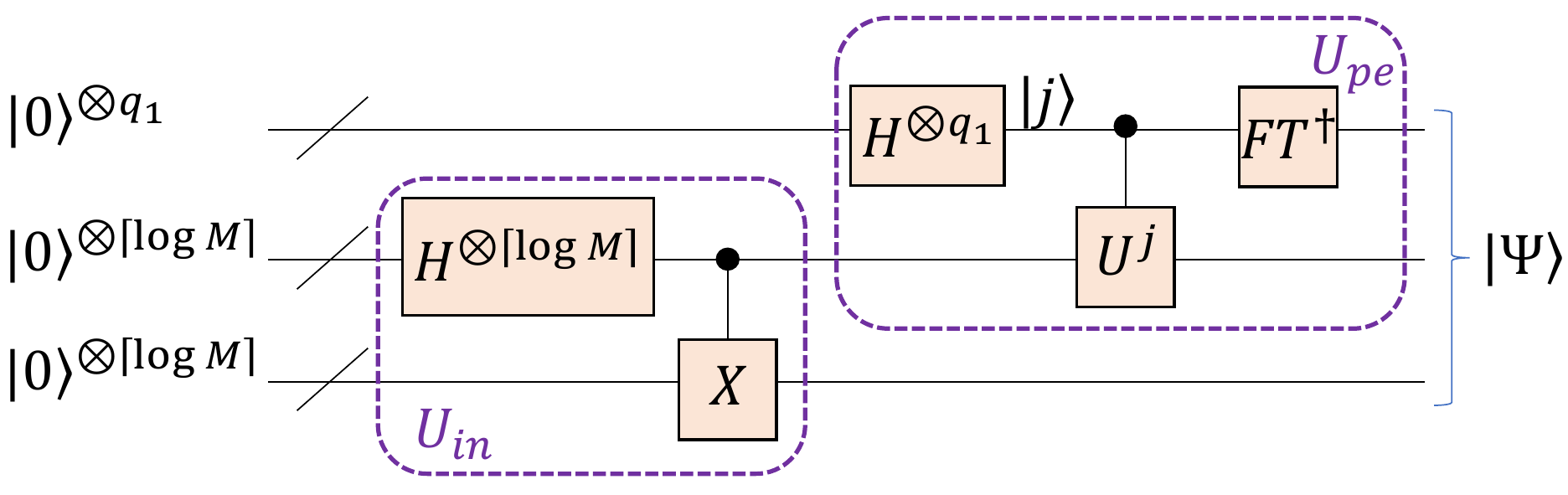}}
	\end{minipage}
	
	\caption{Quantum circuit for preparing $|\Psi\rangle$ of step 2 in the QMEDR framework, where ``/" denotes a bundle of wires, $H$ is the Hadamard gate, $X$ is the NOT gate, $U=e^{\imath(e^{-S_2}e^{S_1})t}$ and $FT$ denotes the quantum Fourier transformation.}
	\label{figure2}
\end{figure}

(2.2) Find the first $m$ smallest eigenvalues.

To invoke quantum minimum-finding algorithm, we should be able to construct an oracle $\textbf{O}_\chi$ to mark the item $\lambda_i\leq\gamma$, where $\gamma$ is a threshold parameter and can be determined by measuring the eigenvalue register. We define a classical Boolean function $\chi$ on the eigenvalue register satisfying
\begin{align}
  \chi(\lambda)=
            \begin{cases}
            1, & \lambda\leq \gamma;\\
0, &  \lambda>\gamma.\\
            \end{cases}
\end{align}
Based on $\chi$, the corresponding quantum oracle $\textbf{O}_\chi$ can be constructed satisfying
\begin{align}
\textbf{O}_\chi|\lambda_i\rangle|{\bf v}_i\rangle|{\bf v}_i\rangle=(-1)^{\chi(\lambda_i)}|\lambda_i\rangle|{\bf v}_i\rangle|{\bf v}_i\rangle.
\end{align}
Then we apply quantum minimum-finding algorithm on $|\Psi\rangle$ to find the minimum eigenvalue, in which the Grover iteration operator is $(2|\Psi\rangle\langle \Psi|-I)\textbf{O}_\chi$.

Without loss of generality,
we assume that the eigenvalues have been arranged in ascending order, that is, $\lambda_0\leq\lambda_1\leq\cdots\leq\lambda_{M-1}$. Suppose that we have obtained the first $s$ smallest eigenvalues. We modify $\chi$ as
\begin{align}
  \chi(\lambda)=
            \begin{cases}
            1, & \lambda\leq \gamma\ \text{and}\ \lambda\notin \{\lambda_j\}_{j=0}^{s-1} ;\\
0, &  \lambda>\gamma,\\
            \end{cases}
\end{align}
and construct the new oracle $\textbf{O}_\chi$ and the new Grover operator. Based on this, we invoke the quantum minimum-finding algorithm again to find $\lambda_s$. The first $m$ smallest eigenvalues $\{\lambda_j\}_{j=0}^{m-1}$ can be obtained in this way. Note that when we get an eigenvalue $\lambda_j$, we also get the corresponding quantum state $|{\bf v}_j\rangle|{\bf v}_j\rangle$.

\subsubsection{ Construct the compressed digital-encoded state}

Since $y_{ij}=\|{\bf x}_i\|\langle {\bf x}_i|{\bf v}_j\rangle$, we can obtain the value of $y_{ij}$ by computing
the inner product $\langle {\bf x}_i|{\bf v}_j\rangle$, $i=0,1,\cdots, N-1$, $j=0,1,\cdots,m-1$.
Without loss of generality, we assume $\langle {\bf x}_i|{\bf v}_j\rangle\geq0$ because both $|{\bf v}_j\rangle$ and $-|{\bf v}_j\rangle$ are eigenvectors of $e^{-S_2}e^{S_1}$ corresponding to the eigenvalue $\lambda_j$. Then we can get the value of $\langle {\bf x}_i|{\bf v}_j\rangle$ by computing $(\langle {\bf x}_i|{\bf v}_j\rangle)^2$.
An intuitive idea is to use the Hadamard test~\cite{LNN2018}. (See Appendix~\ref{app:Hadamardtest} for the reason why we do not compute $\langle {\bf x}_i|{\bf v}_j\rangle$ directly.) However, it would be exhausting to do so over a large dataset. Instead here we use the inner product estimation in Lemma~\ref{lem:inner} to accomplish this task in parallel.

To compute $(\langle {\bf x}_i|{\bf v}_j\rangle)^2$ by Lemma~\ref{lem:inner}, we need two unitary operations to prepare the states $|\phi_{ij}\rangle =|\lambda_j\rangle|{\bf x}_i\rangle|{\bf x}_i\rangle$ and $|\varphi\rangle=\frac{1}{\sqrt{m}}\sum_{j=0}^{m-1}|\lambda_j\rangle|{\bf v}_j\rangle|{\bf v}_j\rangle$ respectively.
For $|\phi_{ij}\rangle$, the following mapping
 \begin{align}\label{equ:mapping}
|i\rangle\frac{1}{\sqrt{m}}\sum_{j=0}^{m-1}|j\rangle|0\rangle\mapsto|i\rangle\frac{1}{\sqrt{m}}\sum_{j=0}^{m-1}|j\rangle|\phi_{ij}\rangle
\end{align}
is performed by $\textbf{O}_2$ and a sequence of controlled unitary operations $CU(j): |j\rangle|0\rangle\mapsto |j\rangle|\lambda_j\rangle$~\cite{Yu2018QPCA}, $j=0,1,\cdots,m-1$. Each $CU(j)$ can be performed efficiently because we have obtained $\{\lambda_j\}_{j=0}^{m-1}$ in step 2.
For $|\varphi\rangle$, we first prepare the state $|\Psi\rangle$ which can be rewritten as
 \begin{align}
\frac{1}{\sqrt{M}}\bigg(\sum_{j=0}^{m-1}|\lambda_j\rangle|{\bf v}_j\rangle|{\bf v}_j\rangle+\sum_{j=m}^{M-1}|\lambda_j\rangle|{\bf v}_j\rangle|{\bf v}_j\rangle\bigg) .
\end{align}
Then, an oracle $O_\lambda$ is defined to mark the items $\{\lambda_j\}_{j=0}^{m-1}$ (similar to ${\bf O}_\chi$), i.e., $O_\lambda|\lambda_j\rangle|{\bf v}_j\rangle|{\bf v}_j\rangle=-|\lambda_j\rangle|{\bf v}_j\rangle|{\bf v}_j\rangle$ for $j=0,1,\cdots,m-1$. With $O_\lambda$, we can use quantum amplitude amplification~\cite{Brassard2002} to get $|\varphi\rangle$ where the Grover operator is $(2|\Psi\rangle\langle \Psi|-I)O_\lambda$.

With the two unitary operations
$|i\rangle\frac{1}{\sqrt{m}}\sum_{j=0}^{m-1}|j\rangle|0\rangle\mapsto|i\rangle\frac{1}{\sqrt{m}}\sum_{j=0}^{m-1}|j\rangle|\phi_{ij}\rangle$
and
$|i\rangle|j\rangle|0\rangle\mapsto|i\rangle|j\rangle|\varphi\rangle$,
we use Hadamard gates and Lemma~\ref{lem:inner} to produce the state
\begin{align}\label{equ:inner1}
\frac{1}{\sqrt{Nm}}\sum_{i=0}^{N-1}\sum_{j=0}^{m-1}|i\rangle|j\rangle|\frac{(\langle {\bf x}_i|{\bf v}_j\rangle)^2}{\sqrt{m}}\rangle|0\rangle^{\otimes q_2}|0\rangle^{\otimes q_2},
\end{align}
where $q_2$ is the largest number of qubits necessary to store $y_{ij}$ and $\|{\bf x}_i\|$.
Then, we perform ${\bf O}_1$ on the first and fourth registers to obtain
\begin{align}\label{equ:o1xi}
\frac{1}{\sqrt{Nm}}\sum_{i=0}^{N-1}\sum_{j=0}^{m-1}|i\rangle|j\rangle|\frac{(\langle {\bf x}_i|{\bf v}_j\rangle)^2}{\sqrt{m}}\rangle|\|{\bf x}_i\|\rangle|0\rangle.
\end{align}
Next, we perform two Quantum Multiply-Adder (QMA)~\cite{QMA,QMA2} on the third and fourth registers to get
\begin{align}
\frac{1}{\sqrt{Nm}}\sum_{i=0}^{N-1}\sum_{j=0}^{m-1}|i\rangle|j\rangle|\frac{y_{ij}^2}{\sqrt{m}}\rangle|\|{\bf x}_i\|\rangle|0\rangle.
\end{align}
The compressed digital-encoded state $|\psi\rangle$ can be obtained by uncomputing the third and fourth registers after applying $U_f: |x\rangle|0\rangle\mapsto|x\rangle|f(x)\rangle$ on the third and fifth registers, where $f(x)=\sqrt{\sqrt{m}x}$ and $x=\frac{y_{ij}^2}{\sqrt{m}}$.

The entire quantum circuit of step 3 is shown in FIG.~\ref{figure3}.

\begin{figure*}[htbp!]
	
	\begin{minipage}{1\linewidth}
		\vspace{3pt}
		\centerline{\includegraphics[width=0.8\textwidth]{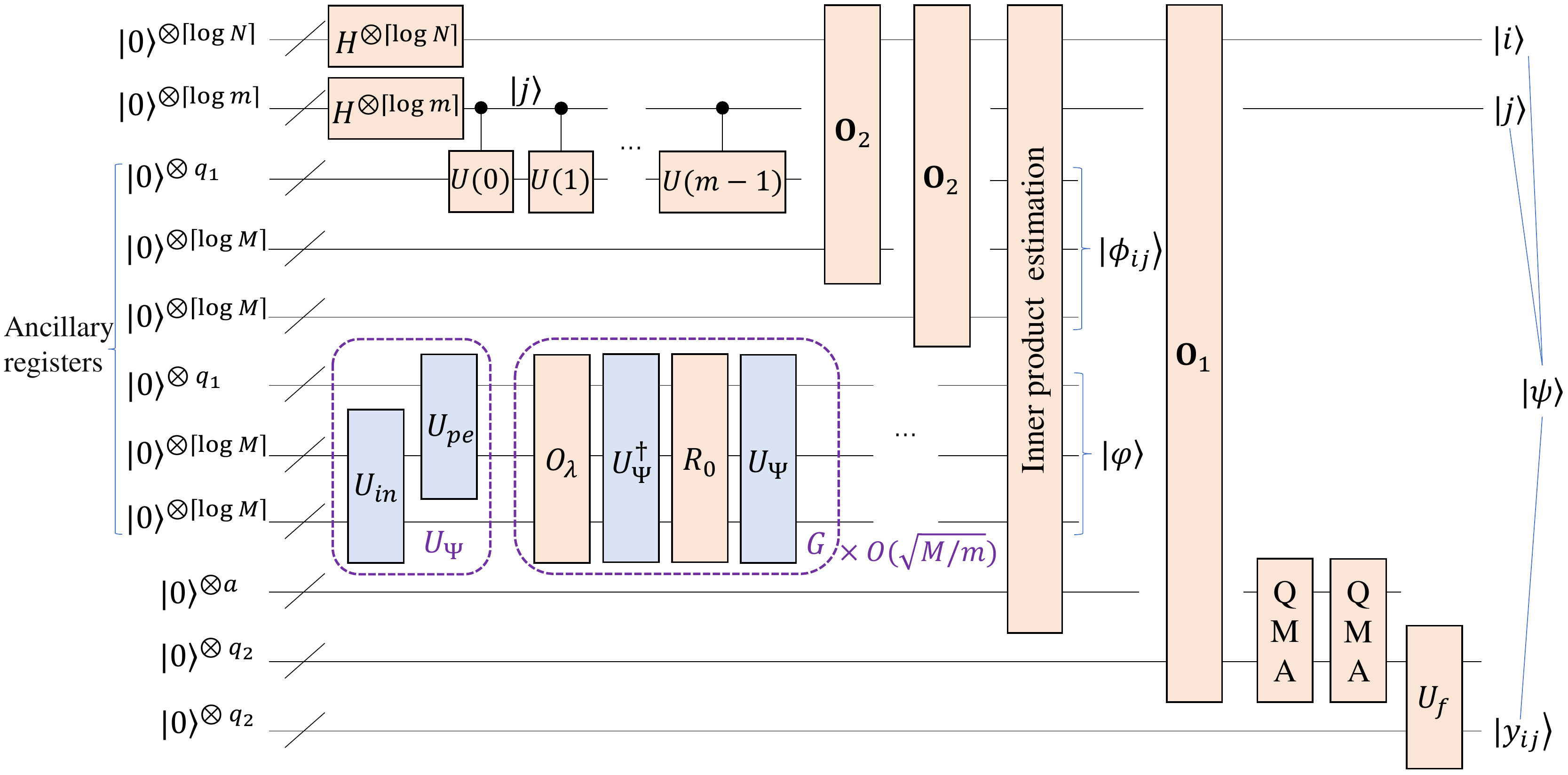}}
	\end{minipage}
	
	\caption{Quantum circuit of step 3 in the QMEDR framework. For simplicity, let $a$ denote the number of qubits required for the inner product estimation, $G$ is the Grover operator in quantum amplitude amplification in which $R_0=2|0\rangle\langle 0|-I$. Note that the ancillary registers are required for inner product estimation and we omit them in Eq.(\ref{equ:inner1}) for convenience. }
	\label{figure3}
\end{figure*}

\subsection{Complexity analysis}\label{sub:ca}

Now we respectively analyze the complexity of each step and discuss the overall complexity.

In step 1, by Lemma~\ref{lem:expH}, the time complexity of implementing the block-encodings of $e^{S_1}$ and $e^{-S_2}$ are $\widetilde{O}(\alpha\kappa_1(a+T_1))$ and $\widetilde{O}(\beta\kappa_2(b+T_2))$ respectively. Then we can implement the block-encoding of $e^{-S_2}e^{S_1}$ in time $\widetilde{O}\big(\max\{\alpha\kappa_1(a+T_1),\beta\kappa_2(b+T_2)\}\big):=T$.

In step 2, for stage (2.1), $O(\lceil\log M\rceil)$ Hadamard gates and CNOT gates are needed to prepare the state in Eq.(\ref{equ:ent}). Next, the quantum phase estimation has a query complexity of $O(\frac{1}{\varepsilon_1}(2+\frac{1}{2\eta}))$ and each query has a time complexity $\widetilde{O}(T+a+b)$, where $\varepsilon_1$ is the error of quantum phase estimation and $1-\eta$ is the probability to succeed. Suppose we wish to approximate $\lambda_j$ to an accuracy $2^{-n}$ with probability of success at least $1-\eta$, we should choose $q_1=n+\lceil\log(2+\frac{1}{\eta})\rceil$~\cite{NC}. For stage (2.2), we should invoke the quantum minimum-finding algorithm for $m$ times to get the first $m$ smallest eigenvalues, and each takes $O(\sqrt{M})$ query complexity.

In step 3, $CU(j)$, $j=0,1,\cdots,m-1$, and two ${\bf O}_2$ are needed to perform the mapping in Eq.(\ref{equ:mapping}). Each $CU(j)$ takes $O(\log(\frac{1}{\varepsilon_1})\log m)$ elementary gates~\cite{NC,Yu2018QPCA}, so $O(m\log(\frac{1}{\varepsilon_1})\log m)$ time is needed to perform all $CU(j)$. Each ${\bf O}_2$ can be done with $\varepsilon_x$-precision in time $O(\text{poly}\log(\frac{NM}{\varepsilon_x}))$. Moreover, the time complexity for preparing the state $|\varphi\rangle$ is $\widetilde{O}(\frac{T+a+b}{\varepsilon_1}\sqrt{\frac{M}{m}})$. Based on the above, the time complexity of producing the state in Eq.(\ref{equ:inner1}) is $\widetilde{O}(\frac{T+a+b}{\varepsilon_1\varepsilon_2}\sqrt{\frac{M}{m}})$, where $\varepsilon_2$ is the error of the inner product estimation. Then, we use ${\bf O}_1$ to get the state in Eq.(\ref{equ:o1xi}) in time $O(\text{poly}\log (NM))$. The time complexity of QMA is $O(\text{poly}\log(\frac{1}{\varepsilon_{Q}}))$ with accuracy $\varepsilon_{Q}$, which can be neglected compared with other subroutines. Moreover, we omit the complexity of $U_f$.

The complexity of each step of the QMEDR framework is summarized as TABLE~\ref{tab:stc}. Furthermore, if every $\lambda_j=O(\frac{1}{m})$, $\varepsilon_1$ should take $O(\frac{1}{m})$ and thus the time complexity of step 2 is $\widetilde{O}((T+a+b)m^2\sqrt{M})$. In step 3, the error of $\frac{(\langle {\bf x}_i|{\bf v}_j\rangle)^2}{\sqrt{m}}$ is $\varepsilon_2$, this will make the error of $y_{ij}^2$ equal to $\sqrt{m}\|{\bf x}_i\|^2\varepsilon_2$. Suppose $y_{ij}=O(1)$, to ensure the final error of $y_{ij}$ is within $\epsilon$, we should take $\varepsilon_2=O(\frac{\epsilon}{\sqrt{m}\max_i\|{\bf x}_i\|^2})$. Therefore, the total time complexity of step 3 is $\widetilde{O}(\frac{(T+a+b)\max_i\|{\bf x}_i\|^2m\sqrt{M}}{\epsilon})$.

\begin{table}
\caption{\label{tab:stc}
The time complexity of each step of the QMEDR framework. Here $\varepsilon_1$ is the error of quantum phase estimation, $\varepsilon_2$ is the error of inner product estimation.}
\begin{ruledtabular}
  \begin{tabular}{cc}
         \specialrule{0em}{1pt}{1pt}
Step  &Time complexity \\ \hline
       \specialrule{0em}{1pt}{1pt}
   step 1&  $T=\widetilde{O}\big(\max\{\alpha\kappa_1(a+T_1),\beta\kappa_2(b+T_2)\}\big)$\\
          \specialrule{0em}{1pt}{1pt}
   step 2&  $\widetilde{O}(\frac{(T+a+b)m\sqrt{M}}{\varepsilon_1})$\\
          \specialrule{0em}{1pt}{1pt}
   step 3&  $\widetilde{O}(\frac{T+a+b}{\varepsilon_1\varepsilon_2}\sqrt{\frac{M}{m}})$\\
          \specialrule{0em}{1pt}{1pt}
  \end{tabular}
  \end{ruledtabular}
\end{table}

As a conclusion, the QMEDR framework can output the compressed digital-encoded state $|\psi\rangle$ in time $\widetilde{O}(\frac{(T+a+b)\max_i\|{\bf x}_i\|^2m\sqrt{M}}{\epsilon})$, where $T=\widetilde{O}\big(\max\{\alpha\kappa_1(a+T_1),\beta\kappa_2(b+T_2)\}\big)$ and $\epsilon$ is the error of $y_{ij}$.

\section{Applications}\label{Sec:exm}

In this section, we will use the QMEDR framework to accelerate the classical ELPP~\cite{ELPP}, EUDP~\cite{UDP,GME2014}, ENPE~\cite{ENPE}, and EDA~\cite{EDA} algorithms which all belong to the MEDR framework. The core is to implement the block-encodings of $S_1$ and $S_2$ of these algorithms. Once the two block-encodings are implemented efficiently, we can design the corresponding quantum algorithms by Theorem~\ref{QMEDR}.

(1) Quantum ELPP algorithm

In ELPP, $S_1=X^TLX$ and $S_2=X^TDX$ where $L=D-S$, $D$ is a diagonal matrix, $D_{ii}=\sum_{j\neq i}S_{ij}$, and the similarity matrix $S$ is defined as
\begin{align}\label{equ:sij}
S_{ij}=
            \begin{cases}
            e^{-\frac{\|{\bf x}_i-{\bf x}_j\|^2}{2\sigma^2}}, & {\bf x}_i\in N_k({\bf x}_j)\  \text{or}\ {\bf x}_j\in N_k({\bf x}_i);\\
0, & \text{otherwise},\\
            \end{cases}
\end{align}
where $\sigma$ is a parameter that is determined empirically and $N_k({\bf x}_j)$ denote the set of the $k$ nearest neighbors of ${\bf x}_j$.

To construct the block-encodings of $S_1$ and $S_2$, we should be able to compute the matrices $S$, $D$, and $L$. We first use the quantum $k$ nearest neighbors algorithm, a straightforward generalization of quantum nearest neighbor classification in~\cite{WiebeKS15}, to obtain the $k$ nearest neighbors of each sample, which takes $\widetilde{O}(kN\sqrt{N})$ time. Then the matrices $S$, $D$, and $L$ can be computed in classical in time $O(kN)$.

For $S_1$, we can implement the block-encodings of $X$ and $L$, then implement its block-encoding by the product of block-encoded matrices~\cite{GSL}. Since $X$ is stored in a structured QRAM~\cite{QRAMDS1,QRAMDS2}, an $(\|X\|_F,\lceil\log(N+M)\rceil,\varepsilon_x)$-block-encoding of $X$ can be implemented in time $O(\text{poly}\log(\frac{NM}{\varepsilon_x}))$ by Lemma 6 in~\cite{ICALP2019}. Since $L$ is a $(k+1)$-sparse matrix, we can implement its a $(k+1,\text{poly}\log(\frac{N}{\varepsilon_l}),\varepsilon_l)$-block-encoding with $O(1)$ queries for the sparse-access oracles and $O(\text{poly}\log(\frac{N}{\varepsilon_l}))$ elementary gates by Lemma 7 in~\cite{ICALP2019}. Then, we implement an $((k+1)\|X\|_F^2,2\lceil\log(N+M)\rceil+\text{poly}\log(\frac{N}{\varepsilon_l}),\|X\|_F^2\varepsilon_l+2(k+1)\|X\|_F\varepsilon_x)$-block-encoding of $S_1$ with complexity $O(\text{poly}\log(\frac{NM}{\varepsilon}))$ where $\varepsilon=\min\{\varepsilon_l,\varepsilon_x\}$~\cite{GSL}.

For the semipositive definite matrix $S_2=\sum_iD_{ii}{\bf x}_i{\bf x}_i^T$, we can construct its block-encoding by preparing the purified density operator $\frac{S_2}{\text{tr}(S_2)} :=\rho$~\cite{GSL}. We first store the vector ${\bf d}=[d_0,d_1,\cdots,d_{N-1}]^T$, $d_i=\sqrt{D_{ii}}$, in a structured QRAM~\cite{QRAMDS1,QRAMDS2}. Note that the time and space complexity of storing ${\bf d}$ are $O(N\log^2N)$ and $O(kN\log^2(kN))$ respectively. Then there exists a quantum algorithm that can perform the mapping $U_d: |i\rangle|0\rangle\mapsto |i\rangle|d_i\rangle$
in time $O(\log N)$. With $U_d$, ${\bf O}_2$ and ${\bf O}_3$, we prepare the state $\frac{\sum_id_i\|{\bf x}_i\|}{\sqrt{\sum_iD_{ii}\|{\bf x}_i\|^2}}|i\rangle|{\bf x}_i\rangle$ whose partial trace over the first register is $\rho$. Then an $(1,O(\log N),\varepsilon_2)$-block-encoding of $\rho$ can be constructed by Lemma 25 in~\cite{GSL}, where $\varepsilon_2$ comes from the unitary operation for preparing the above state. Then we implement a $(\text{tr}(S_2),O(\log N),\varepsilon_2)$-block-encoding of $S_2$ in time $O(\text{poly}\log (\frac{NM}{\varepsilon_x}))$~\cite{Takahira2021}, where the value of $\text{tr}(S_2)$ can be computed in time $O(N)$ and $\text{tr}(S_2)=O(\|X\|_F^2)$.

The quantum ELLP algorithm can then be achieved by Theorem~\ref{QMEDR}.

(2) Quantum EUDP algorithm

In EUDP, $S_1=X^TLX$, $S_2=X^TL'X$ where $L^{'}=D^{'}-S^{'}$, $D^{'}$ is a diagonal matrix, $D^{'}_{ii}=\sum_jS^{'}_{ij}$, and $S^{'}_{ij}=1-S_{ij}$. The block-encoding of $S_1$ is the same as ELPP's. For $S_2$, we first compute the matrix $L'$ in time $O(N^2)$ and store it in a structured QRAM~\cite{QRAMDS1,QRAMDS2}. The space and time complexity to construct the data structure are $O(N^2\log^2N)$. Then we implement an $(\|L'\|_F,\lceil\log(2N)\rceil,\varepsilon_{l'})$ block-encoding of $L'$ and further implement an $(\|X\|_F^2\|L'\|_F,2\lceil\log(N+M)\rceil+\lceil\log(2N)\rceil,\|X\|_F^2\varepsilon_{l'}+2\|X\|_F\|L'\|_F\varepsilon_x)$-block-encoding of $S_2$ in time $O(\text{poly}\log(\frac{NM}{\varepsilon_x}))$ by the product of block-encoded matrices~\cite{GSL}. Then we invoke Theorem~\ref{QMEDR} to obtain the quantum ELLP algorithm.

(3) Quantum ENPE algorithm

In ENPE, $S_1=X^TWX$ and $S_2=X^TX$ where $W=\arg\min\sum_i\|{\bf x}_i-\sum_{{\bf x}_j\in Q({\bf x}_i)}W_{ij}{\bf x}_j\|$ and $\sum_{{\bf x}_j\in Q({\bf x}_i)}W_{ij}=1$. Here we use the $\varepsilon$-neighborhoods criteria to get the nearest neighbors set $Q({\bf x}_i)$ of ${\bf x}_i$, in which each set have $\Theta(k)$ samples. For $S_1$, we first use the method of the quantum NPE algorithm~\cite{QNPE} to get the classical information of the matrix $W$, which can be done in time $\widetilde{O}(N)$. Since $W$ is a matrix of $\Theta(k)$ nonzero elements in each row and column, we assume that it is a $k$-sparse matrix. We can implement a $(k,\text{poly}\log(\frac{N}{\varepsilon_w}),\varepsilon_w)$-block-encoding of $W$ by Lemma 7 in~\cite{ICALP2019} and further implement an $(\|X\|_F^2k,2\lceil\log(N+M)\rceil+\text{poly}\log(\frac{N}{\varepsilon_w}),\|X\|_F^2\varepsilon_w+2\|X\|_Fk\varepsilon_x)$-block-encoding of $S_1$ with complexity $O(\text{poly}\log(\frac{NM}{\varepsilon}))$ where $\varepsilon=\min\{\varepsilon_w,\varepsilon_x\}$~\cite{GSL}. For $S_2$, we can implement an $(\|X\|_F^2,2\lceil\log(N+M)\rceil,2\|X\|_F\varepsilon_x)$-block-encoding of it in time $O(\text{poly}\log(\frac{NM}{\varepsilon_x}))$~\cite{GSL}. The quantum ENPE algorithm is obtained by Theorem~\ref{QMEDR}.

(4) Quantum EDA algorithm

In EDA, $S_1$ and $S_2$ are the between-class scatter matrix and the within-class scatter matrix respectively~\cite{EDA}. According to the quantum LDA algorithm~\cite{QLDA}, we can first construct the density operators corresponding to $S_1$ and $S_2$ in time $O(\log(NM))$. Then an $(1,O(\log M),\varepsilon_1)$-block-encoding of $S_1$ and an $(1,O(\log M),\varepsilon_2)$-block-encoding of $S_2$ are implemented in time $O(\log(NM))$, where $\varepsilon_1$ and $\varepsilon_2$ come from the unitary operations for preparing these two density operators corresponding to $S_1$ and $S_2$. The quantum EDA algorithm can then be achieved by Theorem~\ref{QMEDR}. Note that in quantum EDA, the first $m$ principal eigenvectors are the eigenvectors corresponding to the first $m$ largest eigenvalues of $e^{-S_2}e^{S_1}$. We can replace the quantum minimum-finding algorithm in the step 2 with the quantum maximum-finding algorithm~\cite{1999max}.

\begin{table*}[!htb]
\caption{Complexity comparisons between the classical ELPP, EUDP, ENPE, EDA algorithms and their quantum versions. }\label{tab:comp}
\begin{ruledtabular}
\begin{tabular}{ccccc}
\specialrule{0em}{1pt}{1pt}
Algorithm  &$S_1$ & $S_2$ & Classical complexity & $^b$Quantum complexity \\  \specialrule{0em}{1pt}{1pt} \hline
       \specialrule{0em}{1pt}{1pt}
  ELPP~\cite{ELPP}   &$X^TLX$ & $X^TDX$ & $O(MN^2+M^3)$& $\widetilde{O}(N^{\frac{3}{2}}+T\eta M^{\frac{1}{2}})$, $T\!=\!\widetilde{O}(\max\{\|X\|_F^2\kappa_1, \|X\|_F^2\kappa_2\})$\\
    \specialrule{0em}{1pt}{1pt}
  EUDP~\cite{UDP,GME2014}   &$X^TLX$ & $X^TL^{'}X$ &$O(MN^2+M^3)$ & $\widetilde{O}(N^2+T\eta M^{\frac{1}{2}})$, $T\!=\!\widetilde{O}(\max\{\|X\|_F^2\kappa_1, \|X\|_F^2\|L'\|_F\kappa_2\})$ \\
    \specialrule{0em}{1pt}{1pt}
  ENPE~\cite{ENPE}   &$X^TWX$& $X^TX$ &$O(k^3NM+M^3)$, $k\ll N$~\cite{QNPE} & $\widetilde{O}(kN+T\eta M^{\frac{1}{2}})$, $T\!=\!\widetilde{O}(\max\{\|X\|_F^2\kappa_1,\|X\|_F^2\kappa_2\})$ \\
    \specialrule{0em}{1pt}{1pt}
    $^a$EDA~\cite{EDA}   & $S_b$& $S_w$ &$O(MN^2+N^3)$& $\widetilde{O}(\max\{\kappa_1,\kappa_2\}\eta M^{\frac{1}{2}})$\\
    \specialrule{0em}{1pt}{1pt}
\end{tabular}
\end{ruledtabular}
$^a$In EDA, $S_b$ and $S_w$ are the the between-class scatter matrix and the within-class scatter matrix respectively.
$^b$Here $\kappa_1$ and $\kappa_2$ correspond with $S_1$ and $S_2$ in different algorithms and $\eta=O(\sqrt{m}\max_i\!\|{\bf x}_i\|^2)$. For convenience, we delete the factor $k$ in ELPP and EUDP and let $1/\epsilon, m =O(\text{poly}\log(NM))$ in Theorem~\ref{QMEDR}.
\end{table*}

The complexity comparisons between the classical ELPP, EUDP, ENPE, EDA algorithms and their quantum versions are summarized in TABLE.~\ref{tab:comp}. When $\kappa_1,\kappa_2,\|{\bf x}_i\|=O(\text{poly}\log (NM)$), the results show that the quantum ELPP and quantum NPE algorithms achieve polynomial speedups both in $N$ and $M$, the quantum EUDP algorithm achieve a polynomial speedup in $M$, the quantum EDA algorithm provide an exponential speedup on $N$ and a polynomial speedup on $M$ over their classical counterparts.

\section{Discussion}\label{Sec:dis}

\begin{table*}
\caption{\label{tab1:com}
Comparisons between the QMEDR framework and the existing quantum linear DR algorithms in an end-to-end setting. }
\begin{ruledtabular}
  \begin{tabular}{ccccc}
  \specialrule{0em}{1pt}{1pt}
\multirow{2}{*}{Algorithm}  & \multirow{2}{*}{Input} & &$^a$Output&   \\ \cline{3-5}
\specialrule{0em}{1pt}{1pt}
  & &type I &type II & type III \\
       \specialrule{0em}{1pt}{1pt}\hline
       \specialrule{0em}{1pt}{1pt}
   quantum PCA~\cite{QPCA}&multiple copies of the density operators, $m$ &  $\surd$&& \\
   \specialrule{0em}{1pt}{1pt}
   quantum PCA~\cite{Yu2018QPCA} & $X$ stored in a structured QRAM, $m$ & &$\surd$&   \\
   \specialrule{0em}{1pt}{1pt}
   quantum LDA~\cite{QLDA}     &$X$ stored in a structured QRAM, $m$ & $\surd$ &&    \\
   \specialrule{0em}{1pt}{1pt}
   quantum LDA~\cite{YU2023QLDA}     &$X$ stored in a structured QRAM, $m$ &  &$\surd$ &   \\
   \specialrule{0em}{1pt}{1pt}
   quantum LPP~\cite{QLPP}     &$X$ stored in a structured QRAM, $m$ &$\surd$  & &   \\
   \specialrule{0em}{1pt}{1pt}
   $^b$quantum AOP~\cite{QAOPDuan,QAOP}     & --- &$\surd$ & &  \\
   \specialrule{0em}{1pt}{1pt}
      $^c$quantum NPE~\cite{QNPEliang}   & oracles preparing quantum states $\{|{\bf x}_i\rangle\}_{i=0}^{N-1}$, $m$  & $\surd$ & &   \\
    \specialrule{0em}{1pt}{1pt}
   quantum NPE~\cite{QNPE}     & $X$ stored in a structured QRAM, $m$ & $\surd$ & &   \\
    \specialrule{0em}{1pt}{1pt}
    quantum DCCA~\cite{liDCCA}     & $X$ stored in a QRAM~\cite{QRAM}, $m$ &$\surd$  & &
   \\  \specialrule{0em}{1pt}{1pt}
  The QMEDR framework    &$X$ stored in a structured QRAM, $m$ &   &$\surd$ & $\surd$ \\
   \specialrule{0em}{1pt}{1pt}
  \end{tabular}
  \end{ruledtabular}
  $^a$Here ``type I" denotes that the output is quantum states corresponding to the column vectors of transformation matrix, ``type II" denotes that the output is compressed analog-encoded state, and ``type III" denotes that the output is compressed digital-encoded state. $^b$The quantum AOP algorithms output quantum superposition states corresponding to transformation matrices, here we classify them as ``type I" for convenience. $^c$Although two methods of getting compressed data are given in this algorithm, no superposition compressed state is constructed, so it is classified as ``type I".
\end{table*}

One core of this framework is to construct the block-encoding of the matrix exponential. With it, we can easily combine the block-encoding technique with the quantum phase estimation to reveal the eigenvectors and eigenvalues of $e^{-S_2}e^{S_1}$. The block-encoding framework here we used is a useful tool, which can be applied to algorithms for various problems, such as Hamiltonian simulation and density matrix preparation. By using it, one can significantly improve the existing quantum DR algorithms, such as quantum LPP~\cite{QLPP} and quantum LDA~\cite{QLDA}, and further reduce the dependence of their complexity on error. Moreover, it is useful for constructing the density matrix corresponding to the matrix chain product in form $(A_l\cdots A_2A_1)(A_l\cdots A_2A_1)^\dag$ which can be seen as a special simplified version of Hermitian chain product in~\cite{QLDA}, but here the matrix $A_i$ is not limited to a Hermitian matrix. For the general Hermitian chain product in form $[f_l(A_l)\cdots f_2(A_2)f_1(A_1)][f_l(A_l)\cdots f_2(A_2)f_1(A_1)]^\dag$, the role of block-encoding remains to be explored.

The other core of this framework is to construct the compressed digital-encoded state which can be utilized as input for QML tasks to overcome the curse of dimensionality. For example, quantum $k$-medoids algorithm~\cite{lym} has a time complexity $\widetilde{O}(N^{\frac{1}{2}}M^2)$ for one iteration and is therefore not suitable for dealing with high-dimensional data. One can select a suitable quantum DR algorithm as a preprocessing subroutine to produce the compressed digital-encoded state and input it to the quantum $k$-medoids algorithm. This leads to a better dependence on $M$ in the algorithm complexity.
While the digital-encoded state can be used as input for QML, the analog-encoding is sometimes required, such as quantum support vector machine~\cite{QSVM} and quantum $k$-means clustering~\cite{qmeans2019}. Inspired by~\cite{Yu2018QPCA}, we present the method for constructing the compressed analog-encoded state $\frac{1}{\|Y\|_F}\sum_{i=0}^{N-1}\sum_{j=0}^{m-1}y_{ij}|i\rangle|j\rangle$ in Appendix~\ref{app:analog}. Therefore, our QMEDR framework can output the compressed quantum states in two forms, which can be flexible utilized for various QML tasks.

Now, we divide the quantum linear DR algorithms into three types: (I) output the quantum states corresponding to the column vectors of the transformation matrix; (II) output the compressed analog-encoded state; (III) output the digital-encoded state. The algorithms in type I have not provided the desired quantum data compression, namely obtaining its corresponding low-dimensional dataset. The algorithms in types II and III are well adapted directly to other QML algorithms. The comparisons between the QMEDR framework and the existing quantum linear DR algorithms in an end-to-end setting is summarized as TABLE.~\ref{tab1:com}.

Although the framework is designed for MEDR, with minor modifications it can be used in a variety of linear DR algorithms. For example, for a linear DR algorithm that is related to solving the eigenproblem $A{\bf v}=\lambda {\bf v}$, one can design the block-encoding of $A$ and extract the principal eigenvalues, then step 3 and Appendix~\ref{app:analog} can be used to get the compressed states. Furthermore, for a linear DR algorithm that is related to solving the generalized eigenproblem $A{\bf v}=\lambda B {\bf v}$, one can transform it into the eigenproblem $B^{-\frac{1}{2}}AB^{-\frac{1}{2}}{\bf w}=\lambda {\bf w}$, ${\bf w}=B^{\frac{1}{2}}{\bf v}$~\cite{liDCCA} or the eigenproblem $A^{\frac{1}{2}}B^{-1}A^{\frac{1}{2}}{\bf w}=\lambda {\bf w}$, ${\bf w}=A^{\frac{1}{2}}{\bf v}$~\cite{QLDA}. But in these cases, extra work may be needed, which is worth further exploration.

\section{Conclusion}\label{Sec:con}

In this paper, we proposed the QMEDR framework, which is configurable and a series of new quantum DR algorithms can be derived from this framework. The applications on ELPP, EUDP, ENPE, and EDA showed the quantum superiority of this framework.
The techniques we presented in this paper can be extended to solve many important computational problems, such as computing the matrix exponential and its inverse, simulating the matrix exponential, and performing singular value decomposition (see Appendix~\ref{app:notH2}). Moreover, the QMEDR framework can also be regarded as a common framework to explore the quantization of other linear DR techniques. This work  builds a bridge between quantum linear DR algorithms and other QML algorithms, which is helpful to overcome the curse of dimensionality and solve problems of practical importance. We hope it can inspire the study of QML.

\section*{Acknowledgements}

This work is supported by Beijing Natural Science Foundation (Grant No. 4222031) and National Natural Science Foundation of China (Grant Nos. 61976024, 61972048, 62171056).



\appendix

\section{Proof of Lemma~\ref{lem:expH}}
\label{app:proofexpH}

To prove Lemma~\ref{lem:expH}, we will use the following tools.

\begin{lemma}\label{pro1}
(Block-encoding of controlled-Hamiltonian simulation~\cite{CGJ}).
Let $\mathcal{M}=2^J$ for some $J\in \mathbb{N}$, $\gamma\in \mathbb{R}$ and $\epsilon\geq 0$. Suppose that $U$ is an $(\alpha,a,\frac{\varepsilon}{2(J+1)^2\mathcal{M}\gamma})$-block-encoding of the Hamiltonian $H$.
Then we can implement a $(1,a+2,\varepsilon)$-block-encoding of a controlled $(\mathcal{M},\gamma)$-simulation of the Hamiltonian $H$, with $O(|\alpha\mathcal{M}\gamma|+J\frac{\log (J/\varepsilon)}{\log\log(J/\varepsilon)})$ uses of controlled-$U$ or its inverse and with $O(a|\alpha\mathcal{M}\gamma|+aJ\frac{\log (J/\varepsilon)}{\log\log(J/\varepsilon)})$ three-qubit gates.
\end{lemma}

Note that here the controlled $(\mathcal{M},\gamma)$-simulation of the Hamiltonian $H$ is defined as a unitary
\begin{align}
W :=\sum_{m=-\mathcal{M}}^{\mathcal{M}-1}|m\rangle\langle m|\otimes e^{im\gamma H},
\end{align}
where $|m\rangle$ denotes a (signed) bit string $|b_Jb_{J-1}\cdots b_0\rangle$ such that $m=-b_J2^J+\sum_{j=0}^{J-1}b_j2^J$.

\begin{theorem}\label{pro2}
(Implementing a smooth function of a Hamiltonian~\cite{CGJ}).
Let $x_0\in \mathbb{R}$ and $r>0$ be such that $f(x_0+x)=\sum_{l=0}^{\infty}a_lx^l$ for all $x\in[-r,r]$. Suppose that $B>0$ and $\delta\in(0,r]$ are such that $\sum_{l=0}^{\infty}(r+\delta)^l|a_l|\leq B$. If $\|H-x_0I\|\leq r$ and $\epsilon\in(0,\frac{1}{2}]$, then we can implement a unitary $\tilde{U}$ that is a $(B,b+O(\log(\frac{r\log(1/\epsilon)}{\delta})),B\epsilon)$-block-encoding of $f(H)$, with a single use of a circuit $V$ which is a $(1,b,\frac{\epsilon}{2})$-block-encoding of controlled $(O(\frac{r\log(1/\epsilon)}{\delta}),O(\frac{1}{r}))$-simulation of
$H$, and with using $O(\frac{r}{\delta}\log(\frac{r}{\delta\epsilon})\log(\frac{1}{\epsilon}))$ two-qubit gates.
\end{theorem}

Now we provide the proof of Lemma~\ref{lem:expH}. We first consider the case of $c=1$. Let $f(x) :=e^x$ and observe that
\begin{align}
f(1+x)\!=\!\!\sum_{n=0}^{\infty}\!\frac{(1+x)^n}{n!}\!=\!\!\sum_{n=0}^{\infty}\frac{\sum_{l=0}^\infty\!\tbinom{n}{l} x^l}{n!}
\!=\!\!\sum_{l=0}^{\infty}\frac{\sum_{n=0}^\infty\!\tbinom{n}{l}}{n!} x^l
\end{align}
for all $x\in[-1,1]$, where $\tbinom{n}{l}=\frac{n(n-1)\cdots (n-l+1)}{l!}$. We choose $x_0 :=1$, $r :=1-\frac{1}{\kappa}$, $\delta :=\frac{1}{\kappa}$, and observe that
\begin{align}
\sum_{l=0}^{\infty}(r+\delta)^l\bigg|\frac{\sum_{n=0}^\infty \tbinom{n}{l}}{n!}\bigg|
\!=\!\!\sum_{l=0}^{\infty}\frac{\sum_{n=0}^\infty \tbinom{n}{l}}{n!}
\!=\!\!\sum_{l=0}^{\infty}\frac{2^n}{n!}
\!=\!e^2 \!:=\!B.
\end{align}

Let $\kappa\geq 2$ and $H$ be a Hermitian matrix such that $\frac{I}{\kappa}\preceq H\preceq I$. If $\epsilon\in(0,\frac{1}{2}]$, then we can implement a unitary $\tilde{U}$ that is a $(e^2,b+O(\log((\kappa-1)\log(\frac{1}{\epsilon}))),e^2\epsilon)$-block-encoding of $e^H$, with a single use of a circuit $V$ which is a $(1,b,\frac{\epsilon}{2})$-block-encoding of controlled $(O((\kappa-1)\log(\frac{1}{\epsilon})),O(\frac{\kappa}{\kappa-1}))$-simulation of
$H$, and with using $O((\kappa-1)\log(\frac{\kappa-1}{\epsilon})\log(\frac{1}{\epsilon}))$ two-qubit gates. Let $b :=a+2$, $\frac{\epsilon}{2} :=\varepsilon$, then the circuit $V$ uses
\begin{align}
&O\bigg(|\alpha\kappa\log(\frac{1}{\epsilon})|\nonumber
\\&+\log((\kappa-1)\log(\frac{1}{\epsilon}))\frac{\log \big(\log((\kappa-1)\log(\frac{1}{\epsilon})\big)/\epsilon)}{\log\log\big(\log((\kappa-1)\log(\frac{1}{\epsilon}))/\epsilon\big)}\bigg) \end{align}
controlled-$U$ or its inverse and with
\begin{align}
&O\bigg(a|\alpha\kappa\log(\frac{1}{\epsilon})| \nonumber
\\&+a\log\big((\kappa-1)\log(\frac{1}{\epsilon})\big)\frac{\log \big(\log((\kappa-1)\log(\frac{1}{\epsilon}))/\epsilon\big)}{\log\log\big(\log((\kappa-1)\log(\frac{1}{\epsilon}))/\epsilon\big)}\bigg)
\end{align}
three-qubit gates, where $U$ is a $(\alpha,a,\sigma)$-block-encoding of $H$ and $\sigma =\frac{\epsilon}{4(\log((\kappa-1)\log(1/\epsilon))+1)^2\kappa\log(1/\epsilon)}$.

For simplicity, we delete some items with a small proportion and let $T_U$ denote the cost of $U$. Then the total cost of $\tilde{U}$ is
\begin{equation}
O\bigg(\alpha\kappa\log(\frac{1}{\epsilon})(a+T_U)+(\kappa-1)\log(\frac{\kappa-1}{\epsilon})\log(\frac{1}{\epsilon})\bigg).
\end{equation}

The case of $c=-1$ can be proved similarly.

Then Lemma~\ref{lem:expH} holds.

\section{Implementation of stage (2.1) when $e^{-S_2}e^{S_1}$ is not Hermitian}\label{app:notH2}

If $e^{-S_2}e^{S_1} :=H$ is not Hermitian, we extend it to an embedding Hermitian matrix $\overline{H}=|0\rangle\langle 1|\otimes H + |1\rangle\langle 0|\otimes H^{\dag}$ by the following lemma.
\begin{lemma}\label{lem:2H}
If $U$ is a $(\alpha,a,\delta)$-block-encoding of $H\in \mathbb{R}^{n\times n}$, then $S_{1,a+1}^\dagger(|0\rangle\langle0|\otimes  U +|1\rangle\langle1|\otimes U^\dagger )
(\sigma_x\otimes I_{a+n})S_{1,a+1}$ is an $(\alpha,a,2\delta)$-block-encoding of $\overline{H}=|0\rangle\langle 1|\otimes H + |1\rangle\langle 0|\otimes H^{\dag}\in \mathbb{R}^{2n\times2n}$, where $\sigma_x$ is the Pauli-$X$ gate, $I_w$ is a $w$-qubit identity operator, and $S_{i,j}$ denotes the swap operation for $i$-th and $j$-th qubit.
\end{lemma}
\begin{proof}
Since $S_{1,a+1}(|0\rangle^{\otimes a}\otimes I_{n+1})=I_1\otimes|0\rangle^{\otimes a}\otimes I_{n}$,
we have
\begin{widetext}
\begin{align}
&
\bigg\|\overline{H}-\alpha\big(\langle 0|^{\otimes a}\otimes I_{n+1}\big)S_{1,a+1}^\dagger\big(|0\rangle\langle0|\otimes  U +|1\rangle\langle1|\otimes U^\dagger \big)
\big(\sigma_x\otimes I_{a+n}\big)S_{1,a+1}\big(|0\rangle^{\otimes a}\otimes I_{n+1}\big)\bigg\|\nonumber
\\&
=\bigg\|\overline{H}-\alpha\big(\langle 0|^{\otimes a}\otimes I_{n+1}\big)S_{1,a+1}^\dagger\big(|0\rangle\langle1| \otimes U +|1\rangle\langle0|\otimes U^\dagger \big)
S_{1,a+1}\big(|0\rangle^{\otimes a}\otimes I_{n+1}\big)\bigg\|\nonumber
\\&=\bigg\|\overline{H}-\alpha\big(I_1\otimes \langle 0|^{\otimes a}\otimes I_{n}\big)\big(|0\rangle\langle1| \otimes U +|1\rangle\langle0|\otimes U^\dagger \big)
\big(I_1\otimes|0\rangle^{\otimes a}\otimes I_{n}\big)\bigg\|\nonumber
\\&=\bigg\||0\rangle\langle1|\otimes\bigg(H-\alpha\big(\langle0|^{\otimes a}\otimes I_n\big)U\big(|0\rangle^{\otimes a}\otimes I_n\big)\bigg)
+|1\rangle\langle0|\otimes\bigg(H^\dag-\alpha\big(\langle0|^{\otimes a}\otimes I_n\big)U^\dag\big(|0\rangle^{\otimes a}\otimes I_n\big)\bigg)\bigg\|\nonumber
\\&\leq 2\delta.
\end{align}
\end{widetext}
Then Lemma~\ref{lem:2H} holds.
\end{proof}

$H$ has the singular value decomposition
\begin{align}
H=\sum_{i=0}^{M-1}\lambda_i|{\bf u}_i\rangle\langle {\bf v}_i|,
\end{align}
where $\lambda_i$, $|{\bf u}_i\rangle$ and $|{\bf v}_i\rangle$ are respectively the singular value, the left singular vector, and the right singular vector of $H$. Hence, $\overline{H}$ has $2M$ eigenvalues $\{\pm \lambda_i\}_{i=0}^{M-1}$ and
eigenvectors $\{|{\bf w}_i^{\pm}\rangle\}_{i=0}^{M-1}$, where
\begin{align}
|{\bf w}_i^{\pm}\rangle=\frac{|0\rangle|{\bf u}_i\rangle\pm |1\rangle|{\bf v}_i\rangle}{\sqrt{2}}.
\end{align}
Given the block-encoding of $H$, we can implement the block-encoding of $e^{\imath\overline{H}t}$ by Theorem~\ref{the:OBC} and Lemma~\ref{lem:2H}. By using it, we apply quantum phase estimation on the first two registers of the state in
\begin{align}
|0\rangle^{\otimes q_1}\frac{1}{\sqrt{2M}}\sum_{i=0}^{2M-1}|i\rangle|i\rangle
\end{align}
to obtain the state
 \begin{align}
\frac{1}{\sqrt{2M}}\bigg(\sum_{i=0}^{M-1}|\lambda_i\rangle|{\bf w}_i^+\rangle|{\bf w}_i^+\rangle+\sum_{i=0}^{M-1}|-\lambda_i\rangle|{\bf w}_i^-\rangle|{\bf w}_i^-\rangle\bigg).
\end{align}
Then, for the target state that has the positive value in the eigenvalue register and
has $|1\rangle$ in the first register of $|{\bf w}_i^+\rangle$, we apply quantum amplitude amplification to amplify its amplitude. We can get
 \begin{align}
\frac{1}{\sqrt{M}}\sum_{i=0}^{M-1}|\lambda_i\rangle\bigg(|1\rangle|{\bf v}_i\rangle\bigg)\bigg(|1\rangle|{\bf v}_i\rangle\bigg)
\end{align}
by $O(1)$ times Grover operator iterations. The state $|\Psi\rangle$ is obtained by discarding the second and fourth registers.

\section{Hadamard test fails to compute $\langle {\bf x}_i|{\bf v}_j\rangle$ directly}\label{app:Hadamardtest}

As shown in FIG.~\ref{figure4}, if we want to compute $\langle {\bf x}_i|{\bf v}_j\rangle$, we need two unitaries $U_i: |0\rangle\mapsto|{\bf x}_i\rangle$ and $V_j: |0\rangle\mapsto|{\bf v}_j\rangle$. The $U_i$ is naturally achieved with $\textbf{O}_2$. However, for $V_j$, as discussed in step 2, we are only provided with the unitary $|0\rangle\mapsto|{\bf v}_j\rangle|{\bf v}_j\rangle$ and the additional $|{\bf v}_j\rangle$ cannot be avoided, which will inevitably impact the inner product value. For example, as illustrated in FIG.~\ref{figure5}, with $U_i$ and $\widetilde{V}_j:|0\rangle\mapsto|{\bf v}_j\rangle|{\bf v}_j\rangle$, we can compute the value of $\langle {\bf x}_i|{\bf v}_j\rangle\langle 0|{\bf v}_j\rangle$, but we don't know what the value of $\langle 0|{\bf v}_j\rangle$ is. This is what causes the Hadamard test fail. Moreover, if the measurement in Hadamard test is replaced with parallel amplitude estimation~\cite{Brassard2002,yu2016}, then Lemma~\ref{lem:inner} is derived. Hence, the same reason will lead to the failure of Lemma~\ref{lem:inner} when we utilize it to estimate the inner product $\langle {\bf x}_i|{\bf v}_j\rangle$ directly. Therefore, we only get the value of $\langle {\bf x}_i|{\bf v}_j\rangle$ by computing its square rather than computing it directly.

\begin{figure}[htbp!]
	
	\begin{minipage}{1\linewidth}
		\vspace{3pt}
		\centerline{\includegraphics[width=0.75\textwidth]{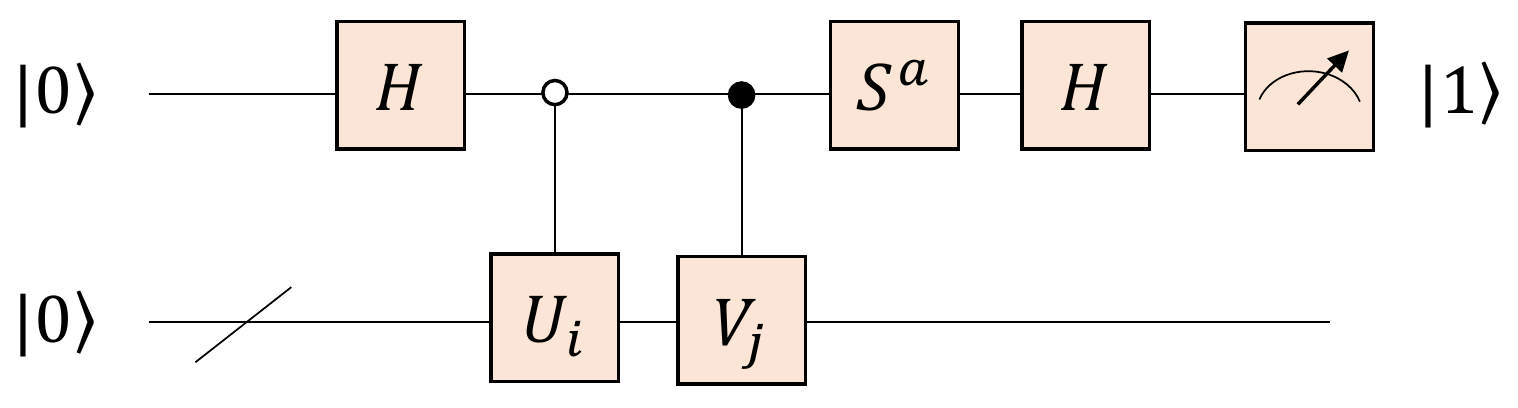}}
	\end{minipage}
	
	\caption{Hadamard test circuit for measuring the real and imaginary part of $\langle \Psi_i|\Phi_j\rangle$ for any arbitrary unitaries $U_i: |0\rangle\mapsto|\Psi_i\rangle$ and $V_j: |0\rangle\mapsto|\Phi_j\rangle$~\cite{LNN2018}. Here $S$ is the phase gate and $a\in\{0,1\}$. The success probability of measuring $|1\rangle$ is given by $P_{ij}=\frac{1-\text{Re}(\zeta\langle \Psi_i|\Phi_j\rangle)}{2}$, where $\zeta=1$ if $a=0$ and $\zeta=\imath$ if $a=1$. Then, $\zeta=1$ and $\zeta=\imath$ recover, respectively, the real and imaginary parts of $\langle \Psi_i|\Phi_j\rangle$.}
	\label{figure4}
\end{figure}
\begin{figure}[htbp!]
	
	\begin{minipage}{1\linewidth}
		\vspace{3pt}
		\centerline{\includegraphics[width=0.65\textwidth]{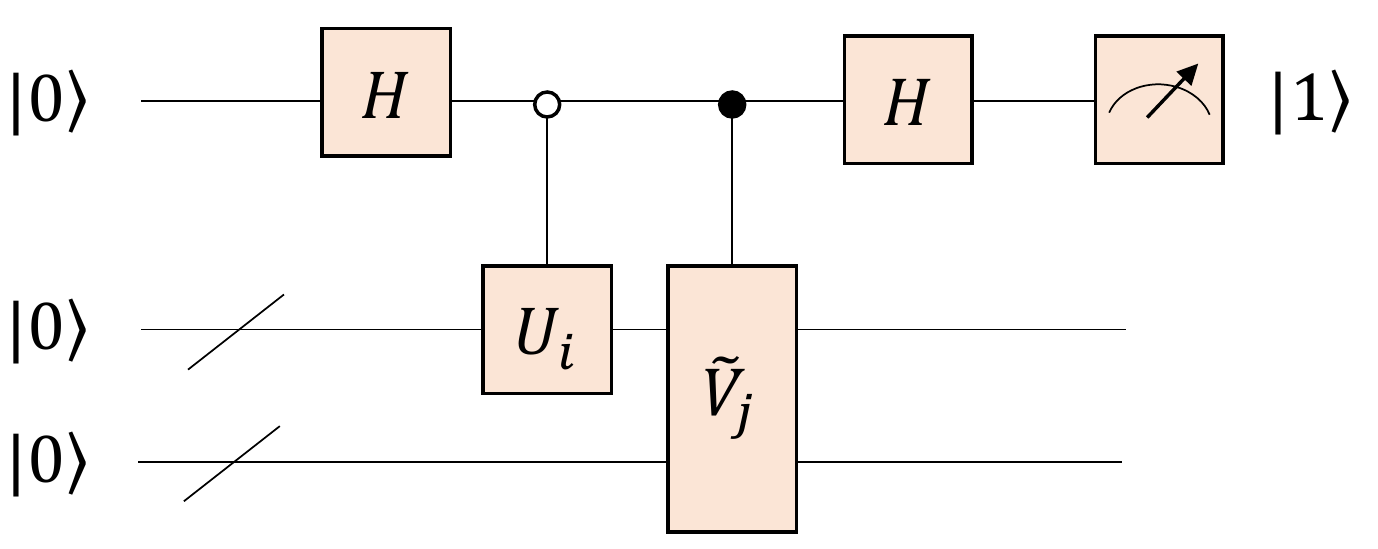}}
	\end{minipage}
	
	\caption{A simple example.}
	\label{figure5}
\end{figure}
	
\begin{lemma}(Distance/Inner products estimation~\cite{qmeans2019}).\label{lem:inner}
Assume that the following unitaries $|i\rangle|0\rangle\mapsto|i\rangle|{\bf x}_i\rangle$ and $|j\rangle|0\rangle\mapsto|j\rangle|{\bf x}_j\rangle$ can be performed in time $T$ and the norms of the vectors are known. For any $\Delta>0$ and $\varepsilon>0$, there exists a quantum algorithm can compute
\begin{align}
|i\rangle|j\rangle|0\rangle\mapsto|i\rangle|j\rangle|{\bf x}_i^T{\bf x}_j\rangle
\end{align}
or
\begin{align}
|i\rangle|j\rangle|0\rangle\mapsto|i\rangle|j\rangle|\|{\bf x}_i-{\bf x}_j\|^2\rangle
\end{align}
with probability at least $1-2\Delta$ for any $\varepsilon$ with complexity $\widetilde{O}(\frac{\|{\bf x}_i\|\|{\bf x}_j\|T\log(1/\Delta)}{\varepsilon})$, where $\varepsilon$ is the error of ${\bf x}_i^T{\bf x}_j$ or $\|{\bf x}_i-{\bf x}_j\|^2$.
\end{lemma}

\section{ Construct the compressed analog-encoded state in the QMEDR framework}\label{app:analog}

\begin{figure*}[htbp!]
	
	\begin{minipage}{1\linewidth}
		\vspace{3pt}
		\centerline{\includegraphics[width=0.8\textwidth]{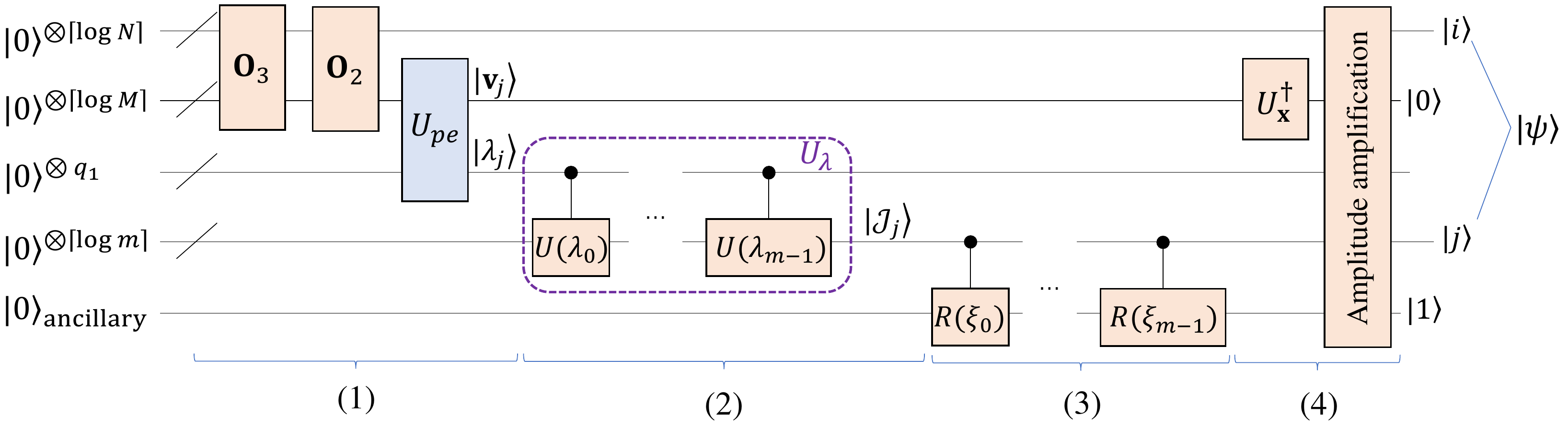}}
	\end{minipage}
	
	\caption{Quantum circuit for preparing the compressed analog-encoded state. $U(\lambda_j)$ and $R(\xi_j)$ correspond with $CU(\lambda_j)$ and $CR(\xi_j)$ respectively, $j=0,1,\cdots,m-1$. $\mathcal{J}_j=j$ for $j=0,1,\cdots, m-1$ and $\mathcal{J}_j=0$ for the others. }
	\label{figure6}
\end{figure*}

This appendix aims to obtain the compressed analog-encoded state
\begin{align}
\frac{1}{\|Y\|_F}\sum_{i=0}^{N-1}\sum_{j=0}^{m-1}y_{ij}|i\rangle|j\rangle :=|\psi\rangle.
\end{align}
Since
\begin{align}
{\bf x}_i=\bigg(\sum_{j=0}^{M-1}|{\bf v}_j\rangle\langle{\bf v}_j|\bigg){\bf x}_i=\sum_{j=0}^{M-1}y_{ij}|{\bf v}_j\rangle,
\end{align}
we have
\begin{align}
\frac{\sum_{i=0}^{N-1}|i\rangle\sum_{j=0}^{M-1}x_{ij}|j\rangle}{\|X\|_F}=\frac{\sum_{i=0}^{N-1}|i\rangle\sum_{j=0}^{M-1}y_{ij}|{\bf v}_j\rangle}{\|X\|_F}.
\end{align}
It implies that we can perform the mapping: $|{\bf v}_j\rangle\mapsto|j\rangle$ on the above state and truncate it to keep the first $m$ terms to get $|\psi\rangle$. But it seems to be difficult. In~\cite{Yu2018QPCA}, the authors introduced an ``anchor state" to achieve the above mapping indirectly. Based on their idea, $|\psi\rangle$ is obtained by the following steps.

(1) We first prepare the initial state
\begin{align}\label{equ:yij}
\frac{\sum_{i=0}^{N-1}\sum_{j=0}^{M-1}x_{ij}|i\rangle|j\rangle|0\rangle^{\otimes q_1}|0\rangle^{\otimes \lceil\log m\rceil} }{\|X\|_F},
\end{align}
which is obtained by using ${\bf O}_2$ and ${\bf O}_3$.
With $e^{\imath(e^{-S_2}e^{S_1})t}$, we perform quantum phase estimation on the second and third registers to get the state
\begin{align}\label{equ:31}
\frac{\sum_{i=0}^{N-1}\sum_{j=0}^{M-1}y_{ij}|i\rangle|{\bf v}_j\rangle|\lambda_j\rangle|0\rangle}{\|X\|_F}.
\end{align}

For simplicity, here we assume that the matrix $e^{-S_2}e^{S_1}$ is Hermitian. If not, we extend it to an embedding Hermitian matrix $\overline{H}$ as in Appendix~\ref{app:notH2} and introduce an ancillary qubit with an initial state of $|1\rangle$. Since $|1\rangle|{\bf v}_j\rangle=\frac{|{\bf w}_j^{+}\rangle-|{\bf w}_j^{-}\rangle}{\sqrt{2}}$, we obtain the new state
\begin{align}
\frac{1}{\|X\|_F}\sum_{i=0}^{N-1}\sum_{j=0}^{M-1}y_{ij}|i\rangle\frac{|{\bf w}_j^{+}\rangle-|{\bf w}_j^{-}\rangle}{\sqrt{2}}|0\rangle|0\rangle.
\end{align}
With $e^{\imath \overline{H}t}$, we perform quantum phase estimation on the second and third registers to get the state
\begin{align}
\frac{1}{\|X\|_F}\sum_{i=0}^{N-1}\sum_{j=0}^{M-1}y_{ij}|i\rangle\frac{|{\bf w}_j^{+}\rangle|\lambda_j\rangle-|{\bf w}_j^{-}\rangle|-\lambda_j\rangle}{\sqrt{2}}|0\rangle.
\end{align}
Next, for the target state that has $|1\rangle$ in the first register of $|{\bf w}_j^+\rangle$ and has the positive value in the third register, we apply quantum amplitude amplification to amplify its amplitude. We can get
\begin{align}
\frac{1}{\|X\|_F}\sum_{i=0}^{N-1}\sum_{j=0}^{M-1}y_{ij}|i\rangle\bigg(|1\rangle|{\bf v}_j\rangle\bigg)|\lambda_j\rangle|0\rangle
\end{align}
by $O(1)$ times Grover operator iterations. The state in Eq.(\ref{equ:31}) is obtained by discarding the ancillary qubit $|1\rangle$.

(2) We perform $U_\lambda$ on the last two registers to get
\begin{align}\label{equ:ind}
&\frac{1}{\|X\|_F}\!\bigg(\!\sum_{i=0}^{N-1}\!\sum_{j=0}^{m-1}\!y_{ij}|i\rangle|{\bf v}_j\!\rangle|\lambda_j\!\rangle|j\rangle
\!+\!\!\sum_{i=0}^{N-1}\!\sum_{j=m}^{M-1}\!y_{ij}|i\rangle|{\bf v}_j\!\rangle|\lambda_j\!\rangle|0\rangle\!\bigg),
\end{align}
where $U_\lambda$ consists of a sequence of controlled unitary operations $CU(\lambda_j):|\lambda_j\rangle|0\rangle\mapsto|\lambda_j\rangle|j\rangle$~\cite{Yu2018QPCA}, $j=0,1,\cdots,m-1$.

(3) We randomly pick out a sample ${\bf x}$ from the dataset $\{{\bf x}_i\}_{i=0}^{N-1}$, whose quantum state $|{\bf x}\rangle$ is prepared by $U_{\bf x}$. Note that $U_{\bf x}$ can be readily implemented via ${\bf O}_2$.
$|{\bf x}\rangle$ will with high probability have a large support in the subspace spanned by the basis $\{|{\bf v}_j\rangle\}_{j=0}^{m-1}$~\cite{Yu2018QPCA}. Then
\begin{align}
|{\bf x}\rangle\approx \xi_0|{\bf v}_0\rangle+\xi_1|{\bf v}_1\rangle+\cdots+\xi_{m-1}|{\bf v}_{m-1}\rangle,
\end{align}
where $\xi_j=\langle {\bf v}_j|{\bf x}\rangle$, $j=0,1,\cdots,m-1$.
Next, we estimate the value of each $\xi_j=\langle {\bf v}_j|{\bf x}\rangle$, which can be done by Hadamard test. With these values, we append an ancillary qubit, and perform $m$ controlled rotations $CR(\xi_j)$, conditioned on $|j\rangle$, to get
\begin{align}
&\frac{1}{\|X\|_F}\bigg[\sum_{i=0}^{N-1}\sum_{j=0}^{m-1}y_{ij}|i\rangle|{\bf v}_j\rangle|\lambda_j\rangle|j\rangle\bigg(\frac{C}{\widehat{\xi}_j}|1\rangle+\sqrt{1-(\frac{C}{\widehat{\xi}_j})^2}|0\rangle\bigg)\nonumber
\\&+\sum_{i=0}^{N-1}\sum_{j=m}^{M-1}y_{ij}|i\rangle|{\bf v}_j\rangle|\lambda_j\rangle|0\rangle|0\rangle\bigg],
\end{align}
where \begin{align}
CR(\xi_j):|j\rangle|0\rangle\mapsto|j\rangle\bigg(\frac{C}{\widehat{\xi}_j}|1\rangle+\sqrt{1-(\frac{C}{\widehat{\xi}_j})^2}|0\rangle\bigg)
\end{align}
and $\widehat{\xi}_j$ is the estimate of $\xi_j$, $C=\min_j\widehat{\xi}_j=O(\frac{1}{\sqrt{m}})$.

(4) We perform $U_{\bf x}^\dag$ on the second register. Next, for the target state having $|0\rangle$ in the second register and having $|1\rangle$ in the last register, we apply quantum amplitude amplification~\cite{Brassard2002} to amplify its amplitude. Then we get $|\psi\rangle$ by discarding the redundant registers.

The entire quantum circuit of constructing $|\psi\rangle$ is shown in FIG.~\ref{figure6}.

We can transform Eq.(\ref{equ:obj}) into a maximum optimization problem, then as analyzed in~\cite{Yu2018QPCA}, the complexity of preparing the compressed analog-encoded state is $O(\text{poly}\log(N,M))$.

\bibliography{refe}

\begin{thebibliography}{64}%
\makeatletter
\providecommand \@ifxundefined [1]{%
 \@ifx{#1\undefined}
}%
\providecommand \@ifnum [1]{%
 \ifnum #1\expandafter \@firstoftwo
 \else \expandafter \@secondoftwo
 \fi
}%
\providecommand \@ifx [1]{%
 \ifx #1\expandafter \@firstoftwo
 \else \expandafter \@secondoftwo
 \fi
}%
\providecommand \natexlab [1]{#1}%
\providecommand \enquote  [1]{``#1''}%
\providecommand \bibnamefont  [1]{#1}%
\providecommand \bibfnamefont [1]{#1}%
\providecommand \citenamefont [1]{#1}%
\providecommand \href@noop [0]{\@secondoftwo}%
\providecommand \href [0]{\begingroup \@sanitize@url \@href}%
\providecommand \@href[1]{\@@startlink{#1}\@@href}%
\providecommand \@@href[1]{\endgroup#1\@@endlink}%
\providecommand \@sanitize@url [0]{\catcode `\\12\catcode `\$12\catcode
  `\&12\catcode `\#12\catcode `\^12\catcode `\_12\catcode `\%12\relax}%
\providecommand \@@startlink[1]{}%
\providecommand \@@endlink[0]{}%
\providecommand \url  [0]{\begingroup\@sanitize@url \@url }%
\providecommand \@url [1]{\endgroup\@href {#1}{\urlprefix }}%
\providecommand \urlprefix  [0]{URL }%
\providecommand \Eprint [0]{\href }%
\providecommand \doibase [0]{http://dx.doi.org/}%
\providecommand \selectlanguage [0]{\@gobble}%
\providecommand \bibinfo  [0]{\@secondoftwo}%
\providecommand \bibfield  [0]{\@secondoftwo}%
\providecommand \translation [1]{[#1]}%
\providecommand \BibitemOpen [0]{}%
\providecommand \bibitemStop [0]{}%
\providecommand \bibitemNoStop [0]{.\EOS\space}%
\providecommand \EOS [0]{\spacefactor3000\relax}%
\providecommand \BibitemShut  [1]{\csname bibitem#1\endcsname}%
\let\auto@bib@innerbib\@empty
\bibitem [{\citenamefont {{Shor}}(1994)}]{PS1994}%
  \BibitemOpen
  \bibfield  {author} {\bibinfo {author} {\bibfnamefont {P.~W.}\ \bibnamefont
  {{Shor}}},\ }in\ \href@noop {} {\emph {\bibinfo {booktitle} {Proceedings 35th
  Annual Symposium on Foundations of Computer Science}}}\ (\bibinfo {year}
  {1994})\ pp.\ \bibinfo {pages} {124--134}\BibitemShut {NoStop}%
\bibitem [{\citenamefont {Grover}(1996)}]{GL1996}%
  \BibitemOpen
  \bibfield  {author} {\bibinfo {author} {\bibfnamefont {L.~K.}\ \bibnamefont
  {Grover}},\ }in\ \href {\doibase 10.1145/237814.237866} {\emph {\bibinfo
  {booktitle} {Proceedings of the Twenty-Eighth Annual ACM Symposium on Theory
  of Computing}}},\ \bibinfo {series and number} {STOC ’96}\ (\bibinfo
  {publisher} {Association for Computing Machinery},\ \bibinfo {address} {New
  York, NY, USA},\ \bibinfo {year} {1996})\ pp.\ \bibinfo {pages}
  {212--219}\BibitemShut {NoStop}%
\bibitem [{\citenamefont {Harrow}\ \emph {et~al.}(2009)\citenamefont {Harrow},
  \citenamefont {Hassidim},\ and\ \citenamefont {Lloyd}}]{HHL}%
  \BibitemOpen
  \bibfield  {author} {\bibinfo {author} {\bibfnamefont {A.~W.}\ \bibnamefont
  {Harrow}}, \bibinfo {author} {\bibfnamefont {A.}~\bibnamefont {Hassidim}}, \
  and\ \bibinfo {author} {\bibfnamefont {S.}~\bibnamefont {Lloyd}},\
  }\href@noop {} {\bibfield  {journal} {\bibinfo  {journal} {Phys. Rev. Lett.}\
  }\textbf {\bibinfo {volume} {103}},\ \bibinfo {pages} {150502} (\bibinfo
  {year} {2009})}\BibitemShut {NoStop}%
\bibitem [{\citenamefont {Li}\ \emph {et~al.}(2022)\citenamefont {Li},
  \citenamefont {Cai}, \citenamefont {Sun}, \citenamefont {Liu}, \citenamefont
  {Wan}, \citenamefont {Qin}, \citenamefont {Wen},\ and\ \citenamefont
  {Gao}}]{lizhenqiang}%
  \BibitemOpen
  \bibfield  {author} {\bibinfo {author} {\bibfnamefont {Z.-Q.}\ \bibnamefont
  {Li}}, \bibinfo {author} {\bibfnamefont {B.-B.}\ \bibnamefont {Cai}},
  \bibinfo {author} {\bibfnamefont {H.-W.}\ \bibnamefont {Sun}}, \bibinfo
  {author} {\bibfnamefont {H.-L.}\ \bibnamefont {Liu}}, \bibinfo {author}
  {\bibfnamefont {L.-C.}\ \bibnamefont {Wan}}, \bibinfo {author} {\bibfnamefont
  {S.-J.}\ \bibnamefont {Qin}}, \bibinfo {author} {\bibfnamefont {Q.-Y.}\
  \bibnamefont {Wen}}, \ and\ \bibinfo {author} {\bibfnamefont
  {F.}~\bibnamefont {Gao}},\ }\href@noop {} {\bibfield  {journal} {\bibinfo
  {journal} {Sci. China Phys. Mech. Astron.}\ }\textbf {\bibinfo {volume} {65}}
  (\bibinfo {year} {2022})}\BibitemShut {NoStop}%
\bibitem [{\citenamefont {Lloyd}\ \emph {et~al.}(2013)\citenamefont {Lloyd},
  \citenamefont {Mohseni},\ and\ \citenamefont {Rebentrost}}]{Lloyd2013}%
  \BibitemOpen
  \bibfield  {author} {\bibinfo {author} {\bibfnamefont {S.}~\bibnamefont
  {Lloyd}}, \bibinfo {author} {\bibfnamefont {M.}~\bibnamefont {Mohseni}}, \
  and\ \bibinfo {author} {\bibfnamefont {P.}~\bibnamefont {Rebentrost}},\
  }\href@noop {} {\bibfield  {journal} {\bibinfo  {journal} {arXiv:1307.0411}\
  } (\bibinfo {year} {2013})}\BibitemShut {NoStop}%
\bibitem [{\citenamefont {Otterbach}\ \emph {et~al.}(2017)\citenamefont
  {Otterbach}, \citenamefont {Manenti}, \citenamefont {Alidoust}, \citenamefont
  {Bestwick}, \citenamefont {Block}, \citenamefont {Bloom}, \citenamefont
  {Caldwell}, \citenamefont {Didier}, \citenamefont {Fried}, \citenamefont
  {Hong}, \citenamefont {Karalekas}, \citenamefont {Osborn}, \citenamefont
  {Papageorge}, \citenamefont {Peterson}, \citenamefont {Prawiroatmodjo},
  \citenamefont {Rubin}, \citenamefont {Ryan}, \citenamefont {Scarabelli},
  \citenamefont {Scheer}, \citenamefont {Sete}, \citenamefont {Sivarajah},
  \citenamefont {Smith}, \citenamefont {Staley}, \citenamefont {Tezak},
  \citenamefont {Zeng}, \citenamefont {Hudson}, \citenamefont {Johnson},
  \citenamefont {Reagor}, \citenamefont {da~Silva},\ and\ \citenamefont
  {Rigetti}}]{Otterbach2017}%
  \BibitemOpen
  \bibfield  {author} {\bibinfo {author} {\bibfnamefont {J.~S.}\ \bibnamefont
  {Otterbach}}, \bibinfo {author} {\bibfnamefont {R.}~\bibnamefont {Manenti}},
  \bibinfo {author} {\bibfnamefont {N.}~\bibnamefont {Alidoust}}, \bibinfo
  {author} {\bibfnamefont {A.}~\bibnamefont {Bestwick}}, \bibinfo {author}
  {\bibfnamefont {M.}~\bibnamefont {Block}}, \bibinfo {author} {\bibfnamefont
  {B.}~\bibnamefont {Bloom}}, \bibinfo {author} {\bibfnamefont
  {S.}~\bibnamefont {Caldwell}}, \bibinfo {author} {\bibfnamefont
  {N.}~\bibnamefont {Didier}}, \bibinfo {author} {\bibfnamefont {E.~S.}\
  \bibnamefont {Fried}}, \bibinfo {author} {\bibfnamefont {S.}~\bibnamefont
  {Hong}}, \bibinfo {author} {\bibfnamefont {P.}~\bibnamefont {Karalekas}},
  \bibinfo {author} {\bibfnamefont {C.~B.}\ \bibnamefont {Osborn}}, \bibinfo
  {author} {\bibfnamefont {A.}~\bibnamefont {Papageorge}}, \bibinfo {author}
  {\bibfnamefont {E.~C.}\ \bibnamefont {Peterson}}, \bibinfo {author}
  {\bibfnamefont {G.}~\bibnamefont {Prawiroatmodjo}}, \bibinfo {author}
  {\bibfnamefont {N.}~\bibnamefont {Rubin}}, \bibinfo {author} {\bibfnamefont
  {C.~A.}\ \bibnamefont {Ryan}}, \bibinfo {author} {\bibfnamefont
  {D.}~\bibnamefont {Scarabelli}}, \bibinfo {author} {\bibfnamefont
  {M.}~\bibnamefont {Scheer}}, \bibinfo {author} {\bibfnamefont {E.~A.}\
  \bibnamefont {Sete}}, \bibinfo {author} {\bibfnamefont {P.}~\bibnamefont
  {Sivarajah}}, \bibinfo {author} {\bibfnamefont {R.~S.}\ \bibnamefont
  {Smith}}, \bibinfo {author} {\bibfnamefont {A.}~\bibnamefont {Staley}},
  \bibinfo {author} {\bibfnamefont {N.}~\bibnamefont {Tezak}}, \bibinfo
  {author} {\bibfnamefont {W.~J.}\ \bibnamefont {Zeng}}, \bibinfo {author}
  {\bibfnamefont {A.}~\bibnamefont {Hudson}}, \bibinfo {author} {\bibfnamefont
  {B.~R.}\ \bibnamefont {Johnson}}, \bibinfo {author} {\bibfnamefont
  {M.}~\bibnamefont {Reagor}}, \bibinfo {author} {\bibfnamefont {M.~P.}\
  \bibnamefont {da~Silva}}, \ and\ \bibinfo {author} {\bibfnamefont
  {C.}~\bibnamefont {Rigetti}},\ }\href {https://arxiv.org/abs/1712.05771}
  {\bibfield  {journal} {\bibinfo  {journal} {arXiv:1712.05771}\ } (\bibinfo
  {year} {2017})}\BibitemShut {NoStop}%
\bibitem [{\citenamefont {Kerenidis}\ and\ \citenamefont
  {Landman}(2021)}]{Kerenidis2021}%
  \BibitemOpen
  \bibfield  {author} {\bibinfo {author} {\bibfnamefont {I.}~\bibnamefont
  {Kerenidis}}\ and\ \bibinfo {author} {\bibfnamefont {J.}~\bibnamefont
  {Landman}},\ }\href {\doibase 10.1103/PhysRevA.103.042415} {\bibfield
  {journal} {\bibinfo  {journal} {Phys. Rev. A}\ }\textbf {\bibinfo {volume}
  {103}},\ \bibinfo {pages} {042415} (\bibinfo {year} {2021})}\BibitemShut
  {NoStop}%
\bibitem [{\citenamefont {Lloyd}\ \emph {et~al.}(2014)\citenamefont {Lloyd},
  \citenamefont {Mohseni},\ and\ \citenamefont {Rebentrost}}]{QPCA}%
  \BibitemOpen
  \bibfield  {author} {\bibinfo {author} {\bibfnamefont {S.}~\bibnamefont
  {Lloyd}}, \bibinfo {author} {\bibfnamefont {M.}~\bibnamefont {Mohseni}}, \
  and\ \bibinfo {author} {\bibfnamefont {P.}~\bibnamefont {Rebentrost}},\
  }\href {\doibase 10.1038/nphys3029} {\bibfield  {journal} {\bibinfo
  {journal} {Nature Physics}\ }\textbf {\bibinfo {volume} {10}},\ \bibinfo
  {pages} {631} (\bibinfo {year} {2014})}\BibitemShut {NoStop}%
\bibitem [{\citenamefont {Cong}\ and\ \citenamefont {Duan}(2016)}]{QLDA}%
  \BibitemOpen
  \bibfield  {author} {\bibinfo {author} {\bibfnamefont {I.}~\bibnamefont
  {Cong}}\ and\ \bibinfo {author} {\bibfnamefont {L.}~\bibnamefont {Duan}},\
  }\href@noop {} {\bibfield  {journal} {\bibinfo  {journal} {New Journal of
  Physics}\ }\textbf {\bibinfo {volume} {18}},\ \bibinfo {pages} {073011}
  (\bibinfo {year} {2016})}\BibitemShut {NoStop}%
\bibitem [{\citenamefont {Pan}\ \emph {et~al.}(2020)\citenamefont {Pan},
  \citenamefont {Wan}, \citenamefont {Liu}, \citenamefont {Wang}, \citenamefont
  {Qin}, \citenamefont {Wen},\ and\ \citenamefont {Gao}}]{QAOP}%
  \BibitemOpen
  \bibfield  {author} {\bibinfo {author} {\bibfnamefont {S.-J.}\ \bibnamefont
  {Pan}}, \bibinfo {author} {\bibfnamefont {L.-C.}\ \bibnamefont {Wan}},
  \bibinfo {author} {\bibfnamefont {H.-L.}\ \bibnamefont {Liu}}, \bibinfo
  {author} {\bibfnamefont {Q.-L.}\ \bibnamefont {Wang}}, \bibinfo {author}
  {\bibfnamefont {S.-J.}\ \bibnamefont {Qin}}, \bibinfo {author} {\bibfnamefont
  {Q.-Y.}\ \bibnamefont {Wen}}, \ and\ \bibinfo {author} {\bibfnamefont
  {F.}~\bibnamefont {Gao}},\ }\href {\doibase 10.1103/PhysRevA.102.052402}
  {\bibfield  {journal} {\bibinfo  {journal} {Phys. Rev. A}\ }\textbf {\bibinfo
  {volume} {102}},\ \bibinfo {pages} {052402} (\bibinfo {year}
  {2020})}\BibitemShut {NoStop}%
\bibitem [{\citenamefont {Pan}\ \emph {et~al.}(2022)\citenamefont {Pan},
  \citenamefont {Wan}, \citenamefont {Liu}, \citenamefont {Wu}, \citenamefont
  {Qin}, \citenamefont {Wen},\ and\ \citenamefont {Gao}}]{QNPE}%
  \BibitemOpen
  \bibfield  {author} {\bibinfo {author} {\bibfnamefont {S.-J.}\ \bibnamefont
  {Pan}}, \bibinfo {author} {\bibfnamefont {L.-C.}\ \bibnamefont {Wan}},
  \bibinfo {author} {\bibfnamefont {H.-L.}\ \bibnamefont {Liu}}, \bibinfo
  {author} {\bibfnamefont {Y.-S.}\ \bibnamefont {Wu}}, \bibinfo {author}
  {\bibfnamefont {S.-J.}\ \bibnamefont {Qin}}, \bibinfo {author} {\bibfnamefont
  {Q.-Y.}\ \bibnamefont {Wen}}, \ and\ \bibinfo {author} {\bibfnamefont
  {F.}~\bibnamefont {Gao}},\ }\href
  {http://iopscience.iop.org/article/10.1088/1674-1056/ac523a} {\bibfield
  {journal} {\bibinfo  {journal} {Chinese Physics B}\ } (\bibinfo {year}
  {2022})}\BibitemShut {NoStop}%
\bibitem [{\citenamefont {Wan}\ \emph {et~al.}(2018)\citenamefont {Wan},
  \citenamefont {Yu}, \citenamefont {Pan}, \citenamefont {Gao}, \citenamefont
  {Wen},\ and\ \citenamefont {Qin}}]{wan2018}%
  \BibitemOpen
  \bibfield  {author} {\bibinfo {author} {\bibfnamefont {L.-C.}\ \bibnamefont
  {Wan}}, \bibinfo {author} {\bibfnamefont {C.-H.}\ \bibnamefont {Yu}},
  \bibinfo {author} {\bibfnamefont {S.-J.}\ \bibnamefont {Pan}}, \bibinfo
  {author} {\bibfnamefont {F.}~\bibnamefont {Gao}}, \bibinfo {author}
  {\bibfnamefont {Q.-Y.}\ \bibnamefont {Wen}}, \ and\ \bibinfo {author}
  {\bibfnamefont {S.-J.}\ \bibnamefont {Qin}},\ }\href {\doibase
  10.1103/PhysRevA.97.062322} {\bibfield  {journal} {\bibinfo  {journal} {Phys.
  Rev. A}\ }\textbf {\bibinfo {volume} {97}},\ \bibinfo {pages} {062322}
  (\bibinfo {year} {2018})}\BibitemShut {NoStop}%
\bibitem [{\citenamefont {Liu}\ \emph {et~al.}(2022{\natexlab{a}})\citenamefont
  {Liu}, \citenamefont {Qin}, \citenamefont {Wan}, \citenamefont {Yu},
  \citenamefont {Pan}, \citenamefont {Gao},\ and\ \citenamefont {Wen}}]{liu}%
  \BibitemOpen
  \bibfield  {author} {\bibinfo {author} {\bibfnamefont {H.-L.}\ \bibnamefont
  {Liu}}, \bibinfo {author} {\bibfnamefont {S.-J.}\ \bibnamefont {Qin}},
  \bibinfo {author} {\bibfnamefont {L.-C.}\ \bibnamefont {Wan}}, \bibinfo
  {author} {\bibfnamefont {C.-H.}\ \bibnamefont {Yu}}, \bibinfo {author}
  {\bibfnamefont {S.-J.}\ \bibnamefont {Pan}}, \bibinfo {author} {\bibfnamefont
  {F.}~\bibnamefont {Gao}}, \ and\ \bibinfo {author} {\bibfnamefont {Q.-Y.}\
  \bibnamefont {Wen}},\ }\href@noop {} {\bibfield  {journal} {\bibinfo
  {journal} {arXiv:2203.14451v1}\ } (\bibinfo {year}
  {2022}{\natexlab{a}})}\BibitemShut {NoStop}%
\bibitem [{\citenamefont {Wan}\ \emph {et~al.}(2021)\citenamefont {Wan},
  \citenamefont {Yu}, \citenamefont {Pan}, \citenamefont {Qin}, \citenamefont
  {Gao},\ and\ \citenamefont {Wen}}]{wan2021}%
  \BibitemOpen
  \bibfield  {author} {\bibinfo {author} {\bibfnamefont {L.-C.}\ \bibnamefont
  {Wan}}, \bibinfo {author} {\bibfnamefont {C.-H.}\ \bibnamefont {Yu}},
  \bibinfo {author} {\bibfnamefont {S.-J.}\ \bibnamefont {Pan}}, \bibinfo
  {author} {\bibfnamefont {S.-J.}\ \bibnamefont {Qin}}, \bibinfo {author}
  {\bibfnamefont {F.}~\bibnamefont {Gao}}, \ and\ \bibinfo {author}
  {\bibfnamefont {Q.-Y.}\ \bibnamefont {Wen}},\ }\href {\doibase
  10.1103/PhysRevA.104.062414} {\bibfield  {journal} {\bibinfo  {journal}
  {Phys. Rev. A}\ }\textbf {\bibinfo {volume} {104}},\ \bibinfo {pages}
  {062414} (\bibinfo {year} {2021})}\BibitemShut {NoStop}%
\bibitem [{\citenamefont {Biamonte}\ \emph {et~al.}(2017)\citenamefont
  {Biamonte}, \citenamefont {Wittek}, \citenamefont {Pancotti}, \citenamefont
  {Rebentrost}, \citenamefont {Wiebe},\ and\ \citenamefont {Lloyd}}]{QML}%
  \BibitemOpen
  \bibfield  {author} {\bibinfo {author} {\bibfnamefont {J.}~\bibnamefont
  {Biamonte}}, \bibinfo {author} {\bibfnamefont {P.}~\bibnamefont {Wittek}},
  \bibinfo {author} {\bibfnamefont {N.}~\bibnamefont {Pancotti}}, \bibinfo
  {author} {\bibfnamefont {P.}~\bibnamefont {Rebentrost}}, \bibinfo {author}
  {\bibfnamefont {N.}~\bibnamefont {Wiebe}}, \ and\ \bibinfo {author}
  {\bibfnamefont {S.}~\bibnamefont {Lloyd}},\ }\href@noop {} {\bibfield
  {journal} {\bibinfo  {journal} {Nature}\ }\textbf {\bibinfo {volume} {549}},\
  \bibinfo {pages} {195} (\bibinfo {year} {2017})}\BibitemShut {NoStop}%
\bibitem [{\citenamefont {Hammer}(1962)}]{Hammer1961}%
  \BibitemOpen
  \bibfield  {author} {\bibinfo {author} {\bibfnamefont {P.~C.}\ \bibnamefont
  {Hammer}},\ }\href {\doibase 10.1137/1004050} {\bibfield  {journal} {\bibinfo
   {journal} {SIAM Review}\ }\textbf {\bibinfo {volume} {4}},\ \bibinfo {pages}
  {163} (\bibinfo {year} {1962})},\ \Eprint
  {http://arxiv.org/abs/https://doi.org/10.1137/1004050}
  {https://doi.org/10.1137/1004050} \BibitemShut {NoStop}%
\bibitem [{\citenamefont {Bishop}(2006)}]{PRML2006}%
  \BibitemOpen
  \bibfield  {author} {\bibinfo {author} {\bibfnamefont {C.}~\bibnamefont
  {Bishop}},\ }\href
  {https://www.microsoft.com/en-us/research/publication/pattern-recognition-machine-learning/}
  {\emph {\bibinfo {title} {Pattern Recognition and Machine Learning}}}\
  (\bibinfo  {publisher} {Springer},\ \bibinfo {year} {2006})\BibitemShut
  {NoStop}%
\bibitem [{\citenamefont {He}\ and\ \citenamefont {Niyogi}(2003)}]{LPP}%
  \BibitemOpen
  \bibfield  {author} {\bibinfo {author} {\bibfnamefont {X.}~\bibnamefont
  {He}}\ and\ \bibinfo {author} {\bibfnamefont {P.}~\bibnamefont {Niyogi}},\
  }in\ \href@noop {} {\emph {\bibinfo {booktitle} {Proceedings of the 16th
  International Conference on Neural Information Processing Systems}}},\
  \bibinfo {series and number} {NIPS'03}\ (\bibinfo  {publisher} {MIT Press},\
  \bibinfo {address} {Cambridge, MA, USA},\ \bibinfo {year} {2003})\ pp.\
  \bibinfo {pages} {153--160}\BibitemShut {NoStop}%
\bibitem [{\citenamefont {Yang}\ \emph {et~al.}(2006)\citenamefont {Yang},
  \citenamefont {Zhang}, \citenamefont {Jin},\ and\ \citenamefont
  {yu~Yang}}]{UDP}%
  \BibitemOpen
  \bibfield  {author} {\bibinfo {author} {\bibfnamefont {J.}~\bibnamefont
  {Yang}}, \bibinfo {author} {\bibfnamefont {D.}~\bibnamefont {Zhang}},
  \bibinfo {author} {\bibfnamefont {Z.}~\bibnamefont {Jin}}, \ and\ \bibinfo
  {author} {\bibfnamefont {J.}~\bibnamefont {yu~Yang}},\ }in\ \href {\doibase
  10.1109/ICPR.2006.1143} {\emph {\bibinfo {booktitle} {18th International
  Conference on Pattern Recognition (ICPR'06)}}},\ Vol.~\bibinfo {volume} {1}\
  (\bibinfo {year} {2006})\ pp.\ \bibinfo {pages} {904--907}\BibitemShut
  {NoStop}%
\bibitem [{\citenamefont {He}\ \emph {et~al.}(2005)\citenamefont {He},
  \citenamefont {Cai}, \citenamefont {Yan},\ and\ \citenamefont {Zhang}}]{NPE}%
  \BibitemOpen
  \bibfield  {author} {\bibinfo {author} {\bibfnamefont {X.}~\bibnamefont
  {He}}, \bibinfo {author} {\bibfnamefont {D.}~\bibnamefont {Cai}}, \bibinfo
  {author} {\bibfnamefont {S.}~\bibnamefont {Yan}}, \ and\ \bibinfo {author}
  {\bibfnamefont {H.-J.}\ \bibnamefont {Zhang}},\ }in\ \href {\doibase
  10.1109/ICCV.2005.167} {\emph {\bibinfo {booktitle} {Tenth IEEE International
  Conference on Computer Vision (ICCV'05) Volume 1}}},\ Vol.~\bibinfo {volume}
  {2}\ (\bibinfo {year} {2005})\ pp.\ \bibinfo {pages} {1208--1213 Vol.
  2}\BibitemShut {NoStop}%
\bibitem [{\citenamefont {Belhumeur}\ \emph {et~al.}(1996)\citenamefont
  {Belhumeur}, \citenamefont {Hespanha},\ and\ \citenamefont {Kriegman}}]{LDA}%
  \BibitemOpen
  \bibfield  {author} {\bibinfo {author} {\bibfnamefont {P.~N.}\ \bibnamefont
  {Belhumeur}}, \bibinfo {author} {\bibfnamefont {J.~P.}\ \bibnamefont
  {Hespanha}}, \ and\ \bibinfo {author} {\bibfnamefont {D.~J.}\ \bibnamefont
  {Kriegman}},\ }\href@noop {} {\bibfield  {journal} {\bibinfo  {journal} {IEEE
  Trans. Pattern Anal. Mach. Intell.}\ }\textbf {\bibinfo {volume} {19}},\
  \bibinfo {pages} {711} (\bibinfo {year} {1996})}\BibitemShut {NoStop}%
\bibitem [{\citenamefont {Yan}\ \emph {et~al.}(2007)\citenamefont {Yan},
  \citenamefont {Xu}, \citenamefont {Zhang}, \citenamefont {Zhang},
  \citenamefont {Yang},\ and\ \citenamefont {Lin}}]{GE2007}%
  \BibitemOpen
  \bibfield  {author} {\bibinfo {author} {\bibfnamefont {S.}~\bibnamefont
  {Yan}}, \bibinfo {author} {\bibfnamefont {D.}~\bibnamefont {Xu}}, \bibinfo
  {author} {\bibfnamefont {B.}~\bibnamefont {Zhang}}, \bibinfo {author}
  {\bibfnamefont {H.-j.}\ \bibnamefont {Zhang}}, \bibinfo {author}
  {\bibfnamefont {Q.}~\bibnamefont {Yang}}, \ and\ \bibinfo {author}
  {\bibfnamefont {S.}~\bibnamefont {Lin}},\ }\href {\doibase
  10.1109/TPAMI.2007.250598} {\bibfield  {journal} {\bibinfo  {journal} {IEEE
  Transactions on Pattern Analysis and Machine Intelligence}\ }\textbf
  {\bibinfo {volume} {29}},\ \bibinfo {pages} {40} (\bibinfo {year}
  {2007})}\BibitemShut {NoStop}%
\bibitem [{\citenamefont {Wang}\ \emph {et~al.}(2014)\citenamefont {Wang},
  \citenamefont {Yan}, \citenamefont {Yang}, \citenamefont {Zhou},\ and\
  \citenamefont {Fu}}]{GME2014}%
  \BibitemOpen
  \bibfield  {author} {\bibinfo {author} {\bibfnamefont {S.-J.}\ \bibnamefont
  {Wang}}, \bibinfo {author} {\bibfnamefont {S.}~\bibnamefont {Yan}}, \bibinfo
  {author} {\bibfnamefont {J.}~\bibnamefont {Yang}}, \bibinfo {author}
  {\bibfnamefont {C.-G.}\ \bibnamefont {Zhou}}, \ and\ \bibinfo {author}
  {\bibfnamefont {X.}~\bibnamefont {Fu}},\ }\href {\doibase
  10.1109/TIP.2013.2297020} {\bibfield  {journal} {\bibinfo  {journal} {IEEE
  Transactions on Image Processing}\ }\textbf {\bibinfo {volume} {23}},\
  \bibinfo {pages} {920} (\bibinfo {year} {2014})}\BibitemShut {NoStop}%
\bibitem [{\citenamefont {Yu}\ \emph {et~al.}(2018)\citenamefont {Yu},
  \citenamefont {Gao}, \citenamefont {Lin},\ and\ \citenamefont
  {Wang}}]{Yu2018QPCA}%
  \BibitemOpen
  \bibfield  {author} {\bibinfo {author} {\bibfnamefont {C.-H.}\ \bibnamefont
  {Yu}}, \bibinfo {author} {\bibfnamefont {F.}~\bibnamefont {Gao}}, \bibinfo
  {author} {\bibfnamefont {S.}~\bibnamefont {Lin}}, \ and\ \bibinfo {author}
  {\bibfnamefont {J.}~\bibnamefont {Wang}},\ }\href@noop {} {\bibfield
  {journal} {\bibinfo  {journal} {Quantum Inf Process}\ }\textbf {\bibinfo
  {volume} {18}} (\bibinfo {year} {2018})}\BibitemShut {NoStop}%
\bibitem [{\citenamefont {He}\ \emph {et~al.}(2019)\citenamefont {He},
  \citenamefont {Sun}, \citenamefont {Lyu},\ and\ \citenamefont
  {Wang}}]{QLLE2019}%
  \BibitemOpen
  \bibfield  {author} {\bibinfo {author} {\bibfnamefont {X.}~\bibnamefont
  {He}}, \bibinfo {author} {\bibfnamefont {L.}~\bibnamefont {Sun}}, \bibinfo
  {author} {\bibfnamefont {C.}~\bibnamefont {Lyu}}, \ and\ \bibinfo {author}
  {\bibfnamefont {X.}~\bibnamefont {Wang}},\ }\href@noop {} {\bibfield
  {journal} {\bibinfo  {journal} {Quantum Information Processing}\ }\textbf
  {\bibinfo {volume} {19}} (\bibinfo {year} {2019})}\BibitemShut {NoStop}%
\bibitem [{\citenamefont {Duan}\ \emph {et~al.}(2019)\citenamefont {Duan},
  \citenamefont {Yuan}, \citenamefont {Xu},\ and\ \citenamefont
  {Li}}]{QAOPDuan}%
  \BibitemOpen
  \bibfield  {author} {\bibinfo {author} {\bibfnamefont {B.-J.}\ \bibnamefont
  {Duan}}, \bibinfo {author} {\bibfnamefont {J.-B.}\ \bibnamefont {Yuan}},
  \bibinfo {author} {\bibfnamefont {J.}~\bibnamefont {Xu}}, \ and\ \bibinfo
  {author} {\bibfnamefont {D.}~\bibnamefont {Li}},\ }\href {\doibase
  10.1103/PhysRevA.99.032311} {\bibfield  {journal} {\bibinfo  {journal} {Phys.
  Rev. A}\ }\textbf {\bibinfo {volume} {99}},\ \bibinfo {pages} {032311}
  (\bibinfo {year} {2019})}\BibitemShut {NoStop}%
\bibitem [{\citenamefont {Liang}\ \emph {et~al.}(2020)\citenamefont {Liang},
  \citenamefont {Shen}, \citenamefont {Li},\ and\ \citenamefont
  {Li}}]{QNPEliang}%
  \BibitemOpen
  \bibfield  {author} {\bibinfo {author} {\bibfnamefont {J.-M.}\ \bibnamefont
  {Liang}}, \bibinfo {author} {\bibfnamefont {S.-Q.}\ \bibnamefont {Shen}},
  \bibinfo {author} {\bibfnamefont {M.}~\bibnamefont {Li}}, \ and\ \bibinfo
  {author} {\bibfnamefont {L.}~\bibnamefont {Li}},\ }\href {\doibase
  10.1103/PhysRevA.101.032323} {\bibfield  {journal} {\bibinfo  {journal}
  {Phys. Rev. A}\ }\textbf {\bibinfo {volume} {101}},\ \bibinfo {pages}
  {032323} (\bibinfo {year} {2020})}\BibitemShut {NoStop}%
\bibitem [{\citenamefont {Li}\ \emph {et~al.}(2020)\citenamefont {Li},
  \citenamefont {Zhou}, \citenamefont {Xu}, \citenamefont {Hu},\ and\
  \citenamefont {Fan}}]{Li2021nDR}%
  \BibitemOpen
  \bibfield  {author} {\bibinfo {author} {\bibfnamefont {Y.}~\bibnamefont
  {Li}}, \bibinfo {author} {\bibfnamefont {R.-G.}\ \bibnamefont {Zhou}},
  \bibinfo {author} {\bibfnamefont {R.}~\bibnamefont {Xu}}, \bibinfo {author}
  {\bibfnamefont {W.}~\bibnamefont {Hu}}, \ and\ \bibinfo {author}
  {\bibfnamefont {P.}~\bibnamefont {Fan}},\ }\href {\doibase
  10.1088/2058-9565/abbe66} {\bibfield  {journal} {\bibinfo  {journal} {Quantum
  Science and Technology}\ }\textbf {\bibinfo {volume} {6}},\ \bibinfo {pages}
  {014001} (\bibinfo {year} {2020})}\BibitemShut {NoStop}%
\bibitem [{\citenamefont {Sornsaeng}\ \emph {et~al.}(2021)\citenamefont
  {Sornsaeng}, \citenamefont {Dangniam}, \citenamefont {Palittapongarnpim},\
  and\ \citenamefont {Chotibut}}]{PRA2021nDR}%
  \BibitemOpen
  \bibfield  {author} {\bibinfo {author} {\bibfnamefont {A.}~\bibnamefont
  {Sornsaeng}}, \bibinfo {author} {\bibfnamefont {N.}~\bibnamefont {Dangniam}},
  \bibinfo {author} {\bibfnamefont {P.}~\bibnamefont {Palittapongarnpim}}, \
  and\ \bibinfo {author} {\bibfnamefont {T.}~\bibnamefont {Chotibut}},\ }\href
  {\doibase 10.1103/PhysRevA.104.052410} {\bibfield  {journal} {\bibinfo
  {journal} {Phys. Rev. A}\ }\textbf {\bibinfo {volume} {104}},\ \bibinfo
  {pages} {052410} (\bibinfo {year} {2021})}\BibitemShut {NoStop}%
\bibitem [{\citenamefont {He}\ \emph {et~al.}(2022)\citenamefont {He},
  \citenamefont {Zhang},\ and\ \citenamefont {Zhao}}]{QLPP}%
  \BibitemOpen
  \bibfield  {author} {\bibinfo {author} {\bibfnamefont {X.-Y.}\ \bibnamefont
  {He}}, \bibinfo {author} {\bibfnamefont {A.-Q.}\ \bibnamefont {Zhang}}, \
  and\ \bibinfo {author} {\bibfnamefont {S.-M.}\ \bibnamefont {Zhao}},\ }\href
  {\doibase https://doi.org/10.1007/s11128-022-03424-w} {\bibfield  {journal}
  {\bibinfo  {journal} {Quantum Inf Process}\ }\textbf {\bibinfo {volume}
  {21}},\ \bibinfo {pages} {86} (\bibinfo {year} {2022})}\BibitemShut {NoStop}%
\bibitem [{\citenamefont {Li}\ \emph {et~al.}(2023{\natexlab{a}})\citenamefont
  {Li}, \citenamefont {Liu}, \citenamefont {Pan}, \citenamefont {Qin},
  \citenamefont {Gao},\ and\ \citenamefont {Wen}}]{liDCCA}%
  \BibitemOpen
  \bibfield  {author} {\bibinfo {author} {\bibfnamefont {Y.-M.}\ \bibnamefont
  {Li}}, \bibinfo {author} {\bibfnamefont {H.-L.}\ \bibnamefont {Liu}},
  \bibinfo {author} {\bibfnamefont {S.-J.}\ \bibnamefont {Pan}}, \bibinfo
  {author} {\bibfnamefont {S.-J.}\ \bibnamefont {Qin}}, \bibinfo {author}
  {\bibfnamefont {F.}~\bibnamefont {Gao}}, \ and\ \bibinfo {author}
  {\bibfnamefont {Q.-Y.}\ \bibnamefont {Wen}},\ }\href
  {https://doi.org/10.1007/s11128-023-03909-2} {\bibfield  {journal} {\bibinfo
  {journal} {Quantum Inf Process}\ }\textbf {\bibinfo {volume} {22}},\ \bibinfo
  {pages} {163} (\bibinfo {year} {2023}{\natexlab{a}})}\BibitemShut {NoStop}%
\bibitem [{\citenamefont {Yu}\ \emph {et~al.}(2023)\citenamefont {Yu},
  \citenamefont {Lin},\ and\ \citenamefont {Guo}}]{YU2023QLDA}%
  \BibitemOpen
  \bibfield  {author} {\bibinfo {author} {\bibfnamefont {K.}~\bibnamefont
  {Yu}}, \bibinfo {author} {\bibfnamefont {S.}~\bibnamefont {Lin}}, \ and\
  \bibinfo {author} {\bibfnamefont {G.-D.}\ \bibnamefont {Guo}},\ }\href
  {\doibase https://doi.org/10.1016/j.physa.2023.128554} {\bibfield  {journal}
  {\bibinfo  {journal} {Physica A: Statistical Mechanics and its Applications}\
  }\textbf {\bibinfo {volume} {614}},\ \bibinfo {pages} {128554} (\bibinfo
  {year} {2023})}\BibitemShut {NoStop}%
\bibitem [{\citenamefont {van Apeldoorn}\ and\ \citenamefont
  {Gily{\'e}n}(2019)}]{ICALP2019SDP}%
  \BibitemOpen
  \bibfield  {author} {\bibinfo {author} {\bibfnamefont {J.}~\bibnamefont {van
  Apeldoorn}}\ and\ \bibinfo {author} {\bibfnamefont {A.}~\bibnamefont
  {Gily{\'e}n}},\ }in\ \href {\doibase 10.4230/LIPIcs.ICALP.2019.99} {\emph
  {\bibinfo {booktitle} {46th International Colloquium on Automata, Languages,
  and Programming (ICALP 2019)}}},\ \bibinfo {series} {Leibniz International
  Proceedings in Informatics (LIPIcs)}, Vol.\ \bibinfo {volume} {132},\
  \bibinfo {editor} {edited by\ \bibinfo {editor} {\bibfnamefont
  {C.}~\bibnamefont {Baier}}, \bibinfo {editor} {\bibfnamefont
  {I.}~\bibnamefont {Chatzigiannakis}}, \bibinfo {editor} {\bibfnamefont
  {P.}~\bibnamefont {Flocchini}}, \ and\ \bibinfo {editor} {\bibfnamefont
  {S.}~\bibnamefont {Leonardi}}}\ (\bibinfo  {publisher} {Schloss
  Dagstuhl--Leibniz-Zentrum fuer Informatik},\ \bibinfo {address} {Dagstuhl,
  Germany},\ \bibinfo {year} {2019})\ pp.\ \bibinfo {pages}
  {99:1--99:15}\BibitemShut {NoStop}%
\bibitem [{\citenamefont {Low}\ and\ \citenamefont {Chuang}(2019)}]{Low2019}%
  \BibitemOpen
  \bibfield  {author} {\bibinfo {author} {\bibfnamefont {G.~H.}\ \bibnamefont
  {Low}}\ and\ \bibinfo {author} {\bibfnamefont {I.~L.}\ \bibnamefont
  {Chuang}},\ }\href {\doibase https://doi.org/10.22331/q-2019-07-12-163}
  {\bibfield  {journal} {\bibinfo  {journal} {Quantum}\ }\textbf {\bibinfo
  {volume} {3}},\ \bibinfo {pages} {163} (\bibinfo {year} {2019})}\BibitemShut
  {NoStop}%
\bibitem [{\citenamefont {Gily{\'e}n}\ \emph {et~al.}(2019)\citenamefont
  {Gily{\'e}n}, \citenamefont {Su}, \citenamefont {Low},\ and\ \citenamefont
  {Wiebe}}]{GSL}%
  \BibitemOpen
  \bibfield  {author} {\bibinfo {author} {\bibfnamefont {A.}~\bibnamefont
  {Gily{\'e}n}}, \bibinfo {author} {\bibfnamefont {Y.}~\bibnamefont {Su}},
  \bibinfo {author} {\bibfnamefont {G.~H.}\ \bibnamefont {Low}}, \ and\
  \bibinfo {author} {\bibfnamefont {N.}~\bibnamefont {Wiebe}},\ }in\ \href@noop
  {} {\emph {\bibinfo {booktitle} {Proceedings of the 51st Annual ACM SIGACT
  Symposium on Theory of Computing}}}\ (\bibinfo {year} {2019})\ pp.\ \bibinfo
  {pages} {193--204}\BibitemShut {NoStop}%
\bibitem [{\citenamefont {Chakraborty}\ \emph {et~al.}(2019)\citenamefont
  {Chakraborty}, \citenamefont {Gily{\'e}n},\ and\ \citenamefont
  {Jeffery}}]{ICALP2019}%
  \BibitemOpen
  \bibfield  {author} {\bibinfo {author} {\bibfnamefont {S.}~\bibnamefont
  {Chakraborty}}, \bibinfo {author} {\bibfnamefont {A.}~\bibnamefont
  {Gily{\'e}n}}, \ and\ \bibinfo {author} {\bibfnamefont {S.}~\bibnamefont
  {Jeffery}},\ }in\ \href {\doibase 10.4230/LIPIcs.ICALP.2019.33} {\emph
  {\bibinfo {booktitle} {46th International Colloquium on Automata, Languages,
  and Programming (ICALP 2019)}}},\ \bibinfo {series} {Leibniz International
  Proceedings in Informatics (LIPIcs)}, Vol.\ \bibinfo {volume} {132},\
  \bibinfo {editor} {edited by\ \bibinfo {editor} {\bibfnamefont
  {C.}~\bibnamefont {Baier}}, \bibinfo {editor} {\bibfnamefont
  {I.}~\bibnamefont {Chatzigiannakis}}, \bibinfo {editor} {\bibfnamefont
  {P.}~\bibnamefont {Flocchini}}, \ and\ \bibinfo {editor} {\bibfnamefont
  {S.}~\bibnamefont {Leonardi}}}\ (\bibinfo  {publisher} {Schloss
  Dagstuhl--Leibniz-Zentrum fuer Informatik},\ \bibinfo {address} {Dagstuhl,
  Germany},\ \bibinfo {year} {2019})\ pp.\ \bibinfo {pages}
  {33:1--33:14}\BibitemShut {NoStop}%
\bibitem [{\citenamefont {Shao}(2020)}]{Shao2020}%
  \BibitemOpen
  \bibfield  {author} {\bibinfo {author} {\bibfnamefont {C.}~\bibnamefont
  {Shao}},\ }\href {\doibase 10.1088/1751-8121/ab5d77} {\bibfield  {journal}
  {\bibinfo  {journal} {Journal of Physics A: Mathematical and Theoretical}\
  }\textbf {\bibinfo {volume} {53}},\ \bibinfo {pages} {045301} (\bibinfo
  {year} {2020})}\BibitemShut {NoStop}%
\bibitem [{\citenamefont {Higham}(2008)}]{Higham2008}%
  \BibitemOpen
  \bibfield  {author} {\bibinfo {author} {\bibfnamefont {N.~J.}\ \bibnamefont
  {Higham}},\ }\href {\doibase 10.1137/1.9780898717778} {\emph {\bibinfo
  {title} {Functions of Matrices}}}\ (\bibinfo  {publisher} {Society for
  Industrial and Applied Mathematics},\ \bibinfo {year} {2008})\ \Eprint
  {http://arxiv.org/abs/https://epubs.siam.org/doi/pdf/10.1137/1.9780898717778}
  {https://epubs.siam.org/doi/pdf/10.1137/1.9780898717778} \BibitemShut
  {NoStop}%
\bibitem [{\citenamefont {Nielsen}\ and\ \citenamefont {Chuang}(2011)}]{NC}%
  \BibitemOpen
  \bibfield  {author} {\bibinfo {author} {\bibfnamefont {M.~A.}\ \bibnamefont
  {Nielsen}}\ and\ \bibinfo {author} {\bibfnamefont {I.~L.}\ \bibnamefont
  {Chuang}},\ }\href@noop {} {\emph {\bibinfo {title} {Quantum Computation and
  Quantum Information: 10th Anniversary Edition}}},\ \bibinfo {edition} {10th}\
  ed.\ (\bibinfo  {publisher} {Cambridge University Press},\ \bibinfo {address}
  {USA},\ \bibinfo {year} {2011})\BibitemShut {NoStop}%
\bibitem [{\citenamefont {Mitarai}\ \emph {et~al.}(2019)\citenamefont
  {Mitarai}, \citenamefont {Kitagawa},\ and\ \citenamefont {Fujii}}]{QADC2019}%
  \BibitemOpen
  \bibfield  {author} {\bibinfo {author} {\bibfnamefont {K.}~\bibnamefont
  {Mitarai}}, \bibinfo {author} {\bibfnamefont {M.}~\bibnamefont {Kitagawa}}, \
  and\ \bibinfo {author} {\bibfnamefont {K.}~\bibnamefont {Fujii}},\ }\href
  {\doibase 10.1103/PhysRevA.99.012301} {\bibfield  {journal} {\bibinfo
  {journal} {Phys. Rev. A}\ }\textbf {\bibinfo {volume} {99}},\ \bibinfo
  {pages} {012301} (\bibinfo {year} {2019})}\BibitemShut {NoStop}%
\bibitem [{\citenamefont {Li}\ \emph {et~al.}(2023{\natexlab{b}})\citenamefont
  {Li}, \citenamefont {Liu}, \citenamefont {Pan}, \citenamefont {Qin},
  \citenamefont {Gao}, \citenamefont {Sun},\ and\ \citenamefont {Wen}}]{lym}%
  \BibitemOpen
  \bibfield  {author} {\bibinfo {author} {\bibfnamefont {Y.-M.}\ \bibnamefont
  {Li}}, \bibinfo {author} {\bibfnamefont {H.-L.}\ \bibnamefont {Liu}},
  \bibinfo {author} {\bibfnamefont {S.-J.}\ \bibnamefont {Pan}}, \bibinfo
  {author} {\bibfnamefont {S.-J.}\ \bibnamefont {Qin}}, \bibinfo {author}
  {\bibfnamefont {F.}~\bibnamefont {Gao}}, \bibinfo {author} {\bibfnamefont
  {D.-X.}\ \bibnamefont {Sun}}, \ and\ \bibinfo {author} {\bibfnamefont
  {Q.-Y.}\ \bibnamefont {Wen}},\ }\href {\doibase 10.1103/PhysRevA.107.022421}
  {\bibfield  {journal} {\bibinfo  {journal} {Phys. Rev. A}\ }\textbf {\bibinfo
  {volume} {107}},\ \bibinfo {pages} {022421} (\bibinfo {year}
  {2023}{\natexlab{b}})}\BibitemShut {NoStop}%
\bibitem [{\citenamefont {Wang}\ \emph {et~al.}(2011)\citenamefont {Wang},
  \citenamefont {Chen}, \citenamefont {Peng},\ and\ \citenamefont
  {Zhou}}]{ELPP}%
  \BibitemOpen
  \bibfield  {author} {\bibinfo {author} {\bibfnamefont {S.-J.}\ \bibnamefont
  {Wang}}, \bibinfo {author} {\bibfnamefont {H.-L.}\ \bibnamefont {Chen}},
  \bibinfo {author} {\bibfnamefont {X.-J.}\ \bibnamefont {Peng}}, \ and\
  \bibinfo {author} {\bibfnamefont {C.-G.}\ \bibnamefont {Zhou}},\ }\href
  {\doibase https://doi.org/10.1016/j.neucom.2011.07.007} {\bibfield  {journal}
  {\bibinfo  {journal} {Neurocomputing}\ }\textbf {\bibinfo {volume} {74}},\
  \bibinfo {pages} {3654} (\bibinfo {year} {2011})}\BibitemShut {NoStop}%
\bibitem [{\citenamefont {Ran}\ \emph {et~al.}(2018)\citenamefont {Ran},
  \citenamefont {Fang},\ and\ \citenamefont {Wu}}]{ENPE}%
  \BibitemOpen
  \bibfield  {author} {\bibinfo {author} {\bibfnamefont {R.}~\bibnamefont
  {Ran}}, \bibinfo {author} {\bibfnamefont {B.}~\bibnamefont {Fang}}, \ and\
  \bibinfo {author} {\bibfnamefont {X.~G.}\ \bibnamefont {Wu}},\ }\href@noop {}
  {\bibfield  {journal} {\bibinfo  {journal} {IEICE Trans. Inf. Syst.}\
  }\textbf {\bibinfo {volume} {101-D}},\ \bibinfo {pages} {1410} (\bibinfo
  {year} {2018})}\BibitemShut {NoStop}%
\bibitem [{\citenamefont {Zhang}\ \emph {et~al.}(2010)\citenamefont {Zhang},
  \citenamefont {Fang}, \citenamefont {Tang}, \citenamefont {Shang},\ and\
  \citenamefont {Xu}}]{EDA}%
  \BibitemOpen
  \bibfield  {author} {\bibinfo {author} {\bibfnamefont {T.}~\bibnamefont
  {Zhang}}, \bibinfo {author} {\bibfnamefont {B.}~\bibnamefont {Fang}},
  \bibinfo {author} {\bibfnamefont {Y.~Y.}\ \bibnamefont {Tang}}, \bibinfo
  {author} {\bibfnamefont {Z.}~\bibnamefont {Shang}}, \ and\ \bibinfo {author}
  {\bibfnamefont {B.}~\bibnamefont {Xu}},\ }\href {\doibase
  10.1109/TSMCB.2009.2024759} {\bibfield  {journal} {\bibinfo  {journal} {IEEE
  Transactions on Systems, Man, and Cybernetics, Part B (Cybernetics)}\
  }\textbf {\bibinfo {volume} {40}},\ \bibinfo {pages} {186} (\bibinfo {year}
  {2010})}\BibitemShut {NoStop}%
\bibitem [{\citenamefont {Moler}\ and\ \citenamefont
  {Van~Loan}(2003)}]{matrixE}%
  \BibitemOpen
  \bibfield  {author} {\bibinfo {author} {\bibfnamefont {C.}~\bibnamefont
  {Moler}}\ and\ \bibinfo {author} {\bibfnamefont {C.}~\bibnamefont
  {Van~Loan}},\ }\href {\doibase 10.1137/S00361445024180} {\bibfield  {journal}
  {\bibinfo  {journal} {SIAM Review}\ }\textbf {\bibinfo {volume} {45}},\
  \bibinfo {pages} {3} (\bibinfo {year} {2003})}\BibitemShut {NoStop}%
\bibitem [{\citenamefont {Chakraborty}\ \emph {et~al.}(2018)\citenamefont
  {Chakraborty}, \citenamefont {Gily{\'e}n},\ and\ \citenamefont
  {Jeffery}}]{CGJ}%
  \BibitemOpen
  \bibfield  {author} {\bibinfo {author} {\bibfnamefont {S.}~\bibnamefont
  {Chakraborty}}, \bibinfo {author} {\bibfnamefont {A.}~\bibnamefont
  {Gily{\'e}n}}, \ and\ \bibinfo {author} {\bibfnamefont {S.}~\bibnamefont
  {Jeffery}},\ }\href@noop {} {\enquote {\bibinfo {title} {The power of
  block-encoded matrix powers: improved regression techniques via faster
  hamiltonian simulation},}\ } (\bibinfo {year} {2018}),\ \Eprint
  {http://arxiv.org/abs/quant-ph/1804.01973v2} {arXiv:quant-ph/1804.01973v2}
  \BibitemShut {NoStop}%
\bibitem [{\citenamefont {Kerenidis}\ and\ \citenamefont
  {Prakash}(2020)}]{KP2020}%
  \BibitemOpen
  \bibfield  {author} {\bibinfo {author} {\bibfnamefont {I.}~\bibnamefont
  {Kerenidis}}\ and\ \bibinfo {author} {\bibfnamefont {A.}~\bibnamefont
  {Prakash}},\ }\href {\doibase 10.1103/PhysRevA.101.022316} {\bibfield
  {journal} {\bibinfo  {journal} {Phys. Rev. A}\ }\textbf {\bibinfo {volume}
  {101}},\ \bibinfo {pages} {022316} (\bibinfo {year} {2020})}\BibitemShut
  {NoStop}%
\bibitem [{\citenamefont {Low}\ and\ \citenamefont {Chuang}(2016)}]{LC2016}%
  \BibitemOpen
  \bibfield  {author} {\bibinfo {author} {\bibfnamefont {G.~H.}\ \bibnamefont
  {Low}}\ and\ \bibinfo {author} {\bibfnamefont {I.~L.}\ \bibnamefont
  {Chuang}},\ }\href@noop {} {\enquote {\bibinfo {title} {Hamiltonian
  simulation by qubitization},}\ } (\bibinfo {year} {2016}),\ \Eprint
  {http://arxiv.org/abs/1610.06546[quant-ph]} {arXiv:1610.06546[quant-ph]}
  \BibitemShut {NoStop}%
\bibitem [{\citenamefont {van Apeldoorn}\ and\ \citenamefont
  {Gily{\'e}n}(2018)}]{SDP}%
  \BibitemOpen
  \bibfield  {author} {\bibinfo {author} {\bibfnamefont {J.}~\bibnamefont {van
  Apeldoorn}}\ and\ \bibinfo {author} {\bibfnamefont {A.}~\bibnamefont
  {Gily{\'e}n}},\ }\href@noop {} {\enquote {\bibinfo {title} {Improvements in
  quantum sdp-solving with applications},}\ } (\bibinfo {year} {2018}),\
  \Eprint {http://arxiv.org/abs/1804.05058[quant-ph]}
  {arXiv:1804.05058[quant-ph]} \BibitemShut {NoStop}%
\bibitem [{\citenamefont {Liu}\ \emph {et~al.}(2022{\natexlab{b}})\citenamefont
  {Liu}, \citenamefont {Yu}, \citenamefont {Wan}, \citenamefont {Qin},
  \citenamefont {Gao},\ and\ \citenamefont {Wen}}]{liuMC}%
  \BibitemOpen
  \bibfield  {author} {\bibinfo {author} {\bibfnamefont {H.-L.}\ \bibnamefont
  {Liu}}, \bibinfo {author} {\bibfnamefont {C.-H.}\ \bibnamefont {Yu}},
  \bibinfo {author} {\bibfnamefont {L.-C.}\ \bibnamefont {Wan}}, \bibinfo
  {author} {\bibfnamefont {S.-J.}\ \bibnamefont {Qin}}, \bibinfo {author}
  {\bibfnamefont {F.}~\bibnamefont {Gao}}, \ and\ \bibinfo {author}
  {\bibfnamefont {Q.}~\bibnamefont {Wen}},\ }\href {\doibase
  https://doi.org/10.1016/j.physa.2022.128227} {\bibfield  {journal} {\bibinfo
  {journal} {Physica A: Statistical Mechanics and its Applications}\ }\textbf
  {\bibinfo {volume} {607}},\ \bibinfo {pages} {128227} (\bibinfo {year}
  {2022}{\natexlab{b}})}\BibitemShut {NoStop}%
\bibitem [{\citenamefont {Kerenidis}\ and\ \citenamefont
  {Prakash}(2017)}]{QRAMDS1}%
  \BibitemOpen
  \bibfield  {author} {\bibinfo {author} {\bibfnamefont {I.}~\bibnamefont
  {Kerenidis}}\ and\ \bibinfo {author} {\bibfnamefont {A.}~\bibnamefont
  {Prakash}},\ }in\ \href {\doibase 10.4230/LIPIcs.ITCS.2017.49} {\emph
  {\bibinfo {booktitle} {8th Innovations in Theoretical Computer Science
  Conference (ITCS 2017)}}},\ \bibinfo {series} {Leibniz International
  Proceedings in Informatics (LIPIcs)}, Vol.~\bibinfo {volume} {67},\ \bibinfo
  {editor} {edited by\ \bibinfo {editor} {\bibfnamefont {C.~H.}\ \bibnamefont
  {Papadimitriou}}}\ (\bibinfo  {publisher} {Schloss Dagstuhl--Leibniz-Zentrum
  fuer Informatik},\ \bibinfo {address} {Dagstuhl, Germany},\ \bibinfo {year}
  {2017})\ pp.\ \bibinfo {pages} {49:1--49:21}\BibitemShut {NoStop}%
\bibitem [{\citenamefont {Wossnig}\ \emph {et~al.}(2018)\citenamefont
  {Wossnig}, \citenamefont {Zhao},\ and\ \citenamefont {Prakash}}]{QRAMDS2}%
  \BibitemOpen
  \bibfield  {author} {\bibinfo {author} {\bibfnamefont {L.}~\bibnamefont
  {Wossnig}}, \bibinfo {author} {\bibfnamefont {Z.}~\bibnamefont {Zhao}}, \
  and\ \bibinfo {author} {\bibfnamefont {A.}~\bibnamefont {Prakash}},\ }\href
  {\doibase 10.1103/PhysRevLett.120.050502} {\bibfield  {journal} {\bibinfo
  {journal} {Phys. Rev. Lett.}\ }\textbf {\bibinfo {volume} {120}},\ \bibinfo
  {pages} {050502} (\bibinfo {year} {2018})}\BibitemShut {NoStop}%
\bibitem [{\citenamefont {Durr}\ and\ \citenamefont {Hoyer}(1996)}]{1996min}%
  \BibitemOpen
  \bibfield  {author} {\bibinfo {author} {\bibfnamefont {C.}~\bibnamefont
  {Durr}}\ and\ \bibinfo {author} {\bibfnamefont {P.}~\bibnamefont {Hoyer}},\
  }\href {https://arxiv.org/abs/quant-ph/9607014} {\bibfield  {journal}
  {\bibinfo  {journal} {arXiv:quant-ph/9607014}\ } (\bibinfo {year}
  {1996})}\BibitemShut {NoStop}%
\bibitem [{\citenamefont {Kerenidis}\ \emph {et~al.}(2019)\citenamefont
  {Kerenidis}, \citenamefont {Landman}, \citenamefont {Luongo},\ and\
  \citenamefont {Prakash}}]{qmeans2019}%
  \BibitemOpen
  \bibfield  {author} {\bibinfo {author} {\bibfnamefont {I.}~\bibnamefont
  {Kerenidis}}, \bibinfo {author} {\bibfnamefont {J.}~\bibnamefont {Landman}},
  \bibinfo {author} {\bibfnamefont {A.}~\bibnamefont {Luongo}}, \ and\ \bibinfo
  {author} {\bibfnamefont {A.}~\bibnamefont {Prakash}},\ }in\ \href
  {https://proceedings.neurips.cc/paper/2019/file/16026d60ff9b54410b3435b403afd226-Paper.pdf}
  {\emph {\bibinfo {booktitle} {Advances in Neural Information Processing
  Systems}}},\ Vol.~\bibinfo {volume} {32},\ \bibinfo {editor} {edited by\
  \bibinfo {editor} {\bibfnamefont {H.}~\bibnamefont {Wallach}}, \bibinfo
  {editor} {\bibfnamefont {H.}~\bibnamefont {Larochelle}}, \bibinfo {editor}
  {\bibfnamefont {A.}~\bibnamefont {Beygelzimer}}, \bibinfo {editor}
  {\bibfnamefont {F.}~\bibnamefont {Alch\'{e}-Buc}}, \bibinfo {editor}
  {\bibfnamefont {E.}~\bibnamefont {Fox}}, \ and\ \bibinfo {editor}
  {\bibfnamefont {R.}~\bibnamefont {Garnett}}}\ (\bibinfo  {publisher} {Curran
  Associates, Inc.},\ \bibinfo {year} {2019})\BibitemShut {NoStop}%
\bibitem [{\citenamefont {Liu}\ and\ \citenamefont
  {Rebentrost}(2018)}]{LNN2018}%
  \BibitemOpen
  \bibfield  {author} {\bibinfo {author} {\bibfnamefont {N.}~\bibnamefont
  {Liu}}\ and\ \bibinfo {author} {\bibfnamefont {P.}~\bibnamefont
  {Rebentrost}},\ }\href {\doibase 10.1103/PhysRevA.97.042315} {\bibfield
  {journal} {\bibinfo  {journal} {Phys. Rev. A}\ }\textbf {\bibinfo {volume}
  {97}},\ \bibinfo {pages} {042315} (\bibinfo {year} {2018})}\BibitemShut
  {NoStop}%
\bibitem [{\citenamefont {Brassard}\ \emph {et~al.}(2002)\citenamefont
  {Brassard}, \citenamefont {Hoyer}, \citenamefont {Mosca},\ and\ \citenamefont
  {Tapp}}]{Brassard2002}%
  \BibitemOpen
  \bibfield  {author} {\bibinfo {author} {\bibfnamefont {G.}~\bibnamefont
  {Brassard}}, \bibinfo {author} {\bibfnamefont {P.}~\bibnamefont {Hoyer}},
  \bibinfo {author} {\bibfnamefont {M.}~\bibnamefont {Mosca}}, \ and\ \bibinfo
  {author} {\bibfnamefont {A.}~\bibnamefont {Tapp}},\ }\href@noop {} {\bibfield
   {journal} {\bibinfo  {journal} {Contemporary Mathematics}\ }\textbf
  {\bibinfo {volume} {305}},\ \bibinfo {pages} {53} (\bibinfo {year}
  {2002})}\BibitemShut {NoStop}%
\bibitem [{\citenamefont {Zhou}\ \emph {et~al.}(2017)\citenamefont {Zhou},
  \citenamefont {Loke}, \citenamefont {Izaac},\ and\ \citenamefont
  {Wang}}]{QMA}%
  \BibitemOpen
  \bibfield  {author} {\bibinfo {author} {\bibfnamefont {S.~S.}\ \bibnamefont
  {Zhou}}, \bibinfo {author} {\bibfnamefont {T.}~\bibnamefont {Loke}}, \bibinfo
  {author} {\bibfnamefont {J.~A.}\ \bibnamefont {Izaac}}, \ and\ \bibinfo
  {author} {\bibfnamefont {J.~B.}\ \bibnamefont {Wang}},\ }\href@noop {}
  {\bibfield  {journal} {\bibinfo  {journal} {Quantum Inf Process}\ }\textbf
  {\bibinfo {volume} {16}},\ \bibinfo {pages} {82} (\bibinfo {year}
  {2017})}\BibitemShut {NoStop}%
\bibitem [{\citenamefont {Ruiz-Perez}\ and\ \citenamefont
  {Garcia-Escartin}(2017)}]{QMA2}%
  \BibitemOpen
  \bibfield  {author} {\bibinfo {author} {\bibfnamefont {L.}~\bibnamefont
  {Ruiz-Perez}}\ and\ \bibinfo {author} {\bibfnamefont {J.~C.}\ \bibnamefont
  {Garcia-Escartin}},\ }\href@noop {} {\bibfield  {journal} {\bibinfo
  {journal} {Quantum Inf Process}\ }\textbf {\bibinfo {volume} {16}},\ \bibinfo
  {pages} {152} (\bibinfo {year} {2017})}\BibitemShut {NoStop}%
\bibitem [{\citenamefont {Wiebe}\ \emph {et~al.}(2015)\citenamefont {Wiebe},
  \citenamefont {Kapoor},\ and\ \citenamefont {Svore}}]{WiebeKS15}%
  \BibitemOpen
  \bibfield  {author} {\bibinfo {author} {\bibfnamefont {N.}~\bibnamefont
  {Wiebe}}, \bibinfo {author} {\bibfnamefont {A.}~\bibnamefont {Kapoor}}, \
  and\ \bibinfo {author} {\bibfnamefont {K.~M.}\ \bibnamefont {Svore}},\ }\href
  {http://www.rintonpress.com/xxqic15/qic-15-34/0316-0356.pdf} {\bibfield
  {journal} {\bibinfo  {journal} {Quantum Information \& Computation}\ }\textbf
  {\bibinfo {volume} {15}},\ \bibinfo {pages} {316} (\bibinfo {year}
  {2015})}\BibitemShut {NoStop}%
\bibitem [{\citenamefont {Takahira}\ \emph {et~al.}(2021)\citenamefont
  {Takahira}, \citenamefont {Ohashi}, \citenamefont {Sogabe},\ and\
  \citenamefont {Usuda}}]{Takahira2021}%
  \BibitemOpen
  \bibfield  {author} {\bibinfo {author} {\bibfnamefont {S.}~\bibnamefont
  {Takahira}}, \bibinfo {author} {\bibfnamefont {A.}~\bibnamefont {Ohashi}},
  \bibinfo {author} {\bibfnamefont {T.}~\bibnamefont {Sogabe}}, \ and\ \bibinfo
  {author} {\bibfnamefont {T.~S.}\ \bibnamefont {Usuda}},\ }\href@noop {}
  {\bibfield  {journal} {\bibinfo  {journal} {Quantum Inf. Comput.}\ }\textbf
  {\bibinfo {volume} {22}},\ \bibinfo {pages} {965} (\bibinfo {year}
  {2021})}\BibitemShut {NoStop}%
\bibitem [{\citenamefont {Ahuja}\ and\ \citenamefont {Kapoor}(1999)}]{1999max}%
  \BibitemOpen
  \bibfield  {author} {\bibinfo {author} {\bibfnamefont {A.}~\bibnamefont
  {Ahuja}}\ and\ \bibinfo {author} {\bibfnamefont {S.}~\bibnamefont {Kapoor}},\
  }\href@noop {} {\bibfield  {journal} {\bibinfo  {journal}
  {arXiv:quant-ph/9911082v1}\ } (\bibinfo {year} {1999})}\BibitemShut {NoStop}%
\bibitem [{\citenamefont {Giovannetti}\ \emph {et~al.}(2008)\citenamefont
  {Giovannetti}, \citenamefont {Lloyd},\ and\ \citenamefont {Maccone}}]{QRAM}%
  \BibitemOpen
  \bibfield  {author} {\bibinfo {author} {\bibfnamefont {V.}~\bibnamefont
  {Giovannetti}}, \bibinfo {author} {\bibfnamefont {S.}~\bibnamefont {Lloyd}},
  \ and\ \bibinfo {author} {\bibfnamefont {L.}~\bibnamefont {Maccone}},\
  }\href@noop {} {\bibfield  {journal} {\bibinfo  {journal} {Phys. Rev. Lett.}\
  }\textbf {\bibinfo {volume} {100}},\ \bibinfo {pages} {160501} (\bibinfo
  {year} {2008})}\BibitemShut {NoStop}%
\bibitem [{\citenamefont {Rebentrost}\ \emph {et~al.}(2014)\citenamefont
  {Rebentrost}, \citenamefont {Mohseni},\ and\ \citenamefont {Lloyd}}]{QSVM}%
  \BibitemOpen
  \bibfield  {author} {\bibinfo {author} {\bibfnamefont {P.}~\bibnamefont
  {Rebentrost}}, \bibinfo {author} {\bibfnamefont {M.}~\bibnamefont {Mohseni}},
  \ and\ \bibinfo {author} {\bibfnamefont {S.}~\bibnamefont {Lloyd}},\ }\href
  {\doibase 10.1103/PhysRevLett.113.130503} {\bibfield  {journal} {\bibinfo
  {journal} {Phys. Rev. Lett.}\ }\textbf {\bibinfo {volume} {113}},\ \bibinfo
  {pages} {130503} (\bibinfo {year} {2014})}\BibitemShut {NoStop}%
\bibitem [{\citenamefont {Yu}\ \emph {et~al.}(2016)\citenamefont {Yu},
  \citenamefont {Gao}, \citenamefont {Wang},\ and\ \citenamefont
  {Wen}}]{yu2016}%
  \BibitemOpen
  \bibfield  {author} {\bibinfo {author} {\bibfnamefont {C.-H.}\ \bibnamefont
  {Yu}}, \bibinfo {author} {\bibfnamefont {F.}~\bibnamefont {Gao}}, \bibinfo
  {author} {\bibfnamefont {Q.-L.}\ \bibnamefont {Wang}}, \ and\ \bibinfo
  {author} {\bibfnamefont {Q.-Y.}\ \bibnamefont {Wen}},\ }\href {\doibase
  10.1103/PhysRevA.94.042311} {\bibfield  {journal} {\bibinfo  {journal} {Phys.
  Rev. A}\ }\textbf {\bibinfo {volume} {94}},\ \bibinfo {pages} {042311}
  (\bibinfo {year} {2016})}\BibitemShut {NoStop}%
\end{thebibliography}%
\end{document}